\newif\iffull
\fulltrue 

\documentclass[letterpaper,11pt]{article}

\usepackage{typearea}
\typearea{14}

\usepackage[linktocpage=true,
	pagebackref=true,colorlinks,
	linkcolor=blue,citecolor=blue]
	{hyperref}

\usepackage{amsmath,amssymb,amsthm}
\usepackage{thmtools,thm-restate}
\usepackage{nicefrac}
\usepackage{bbm}
\usepackage{xspace}
\usepackage{color}

\newtheorem{theorem}{Theorem}[section]
\newtheorem{lemma}[theorem]{Lemma}

\newtheorem{claim}[theorem]{Claim}

\newtheorem{remark}{Remark}
\newtheorem{assumption}[theorem]{Assumption}

\theoremstyle{definition}


\usepackage{algorithm}
\usepackage[noend]{algpseudocode}

\def\E{\mathbb{E}}

\def\N{\mathbb{N}}

\def\P{\mathbb{P}}

\def\R{\mathbb{R}}

\def\e{\varepsilon}
\def\eps{\varepsilon}

\def\cD{\mathcal{D}}
\def\cE{\mathcal{E}}
\def\cF{\mathcal{F}}

\def\cL{\mathcal{L}}

\def\Var{\text{Var}}

\newcommand{\ip}[1]{{\langle #1 \rangle}}

\def\sse{\subseteq}

\renewcommand{\bar}[1]{\overline{#1}}

\renewcommand{\emptyset}{\varnothing}

\newcommand{\blue}[1]{{\color{blue} #1}}
\newcommand{\red}[1]{{\color{red} #1}}

\newcommand{\ones}{\mathbbm{1}_d}
\newcommand{\one}{\mathbf{1}}
\newcommand{\LOPT}{\mathcal{L}}

\DeclareMathOperator{\poly}{poly}
\DeclareMathOperator{\alg}{Alg}
\DeclareMathOperator{\OPT}{OPT}
\DeclareMathOperator{\stoch}{\Gitems_T}
\DeclareMathOperator{\items}{items}
\DeclareMathOperator{\argmax}{argmax}
\DeclareMathOperator{\median}{median}

\newcommand{\Gitems}{\mathcal{G}}
\newcommand{\Ritems}{\mathcal{R}}
\newcommand{\Gsize}{G}
\newcommand{\Rsize}{R}
\newcommand{\parts}{K}
\newcommand{\val}{\mathsf{val}}
\newcommand{\nf}{\nicefrac}
\newcommand{\IGNORE}[1]{}

\renewcommand{\ip}[2]{\langle #1,#2 \rangle}
\newcommand{\Bigip}[2]{{\Big\langle #1,#2 \Big\rangle}}

\newcounter{note}[section]

\newcommand{\ts}{\textstyle}

\def\Struct{\textsc{Struct}\xspace}
\def\Search{\textsc{Search}\xspace}
\def\IBGSZ{\textsc{I-bgsz}}
\newcommand{\algx}{{x}}
\newcommand{\dual}{{\lambda}}

\newcommand{\gmax}{g_{\max}}

\newcommand{\tsum}{{\textstyle \sum}}

\newcommand{\initOneLiners}{%
    \setlength{\itemsep}{0pt}
    \setlength{\parsep }{0pt}
    \setlength{\topsep }{0pt}
}
\newenvironment{OneLiners}[1][\ensuremath{\bullet}]
    {\begin{list}
        {#1}
        {\initOneLiners}}
    {\end{list}}

\newcommand{\Cref}[1]{\ref{#1}}

\usepackage{cleveref}
\usepackage{nicefrac}

\Crefname{claim}{Claim}{Claims}
\Crefname{@theorem}{Theorem}{Theorems}

\begin{document}

\title{ Robust Secretary and Prophet Algorithms\\ for Packing
  Integer Programs\thanks{CJA and AG were supported in part by NSF awards CCF-1907820, CCF1955785, and CCF-2006953. MM was partially supported by the Coordena\c{c}\~ao de Aperfei\c{c}oamento de Pessoal de Nivel Superior - Brasil (CAPES) - Finance Code 001, CNPq Bolsa de Produtividade em Pesquisa $\#4$310516/2017-0 and FAPERJ grant Jovem Cientista do Nosso Estado.}}
\author{C.J. Argue\thanks{Carnegie Mellon University, cargue@andrew.cmu.edu.}
\and Anupam Gupta\thanks{Carnegie Mellon University, anupamg@cs.cmu.edu.}
\and Marco Molinaro\thanks{PUC Rio, marco.molinaro@gmail.com.}
\and Sahil Singla\thanks{Georgia Tech, ssingla@gatech.edu.}}


\maketitle



\begin{abstract}
\medskip 
  We study the problem of solving Packing Integer Programs (PIPs) in the
  online setting, where columns  in $[0,1]^d$ of the constraint matrix are revealed
  sequentially, and the goal is to pick a subset of the columns that sum
  to at most $B$ in each coordinate while maximizing the objective. Excellent results are known in
  the secretary setting, where the columns are adversarially chosen,
  but presented in a uniformly random
  order.
  However, these existing algorithms are susceptible to adversarial attacks:
  they try to ``learn'' characteristics of a good solution, but
 tend to over-fit to the model, and hence a small number of adversarial corruptions
  can  cause the algorithm to fail.

  In this paper, we give the first robust algorithms for Packing
  Integer Programs, specifically in the recently proposed Byzantine
  Secretary framework~\cite{BGSZ-ITCS20}. Our techniques are based on a two-level use of
  online learning, to robustly learn an approximation to the optimal
  value, and then to use this robust estimate to pick a good solution.
  These techniques are general and we use them to design robust algorithms for PIPs in the
  prophet model as well, specifically in the
  Prophet-with-Augmentations framework~\cite{ISW-EC20}.  We also improve known results in the Byzantine
  Secretary framework: we make the
  non-constructive results algorithmic and improve the existing bounds
  for single-item and matroid constraints.
\end{abstract}


\section{Introduction} \label{sec:intro}

Resource allocation is a central problem in online decision making:
here, a set of requests for resources arrive one-by-one, each having
an associated value. The goal is to accept a subset of the requests with a large
total value, subject to satisfying given resource constraints. In online
algorithms, we have to make these decisions sequentially and
irrevocably, without the knowledge of future requests. It is common to
model these problems as \emph{packing integer programs} (PIPs) of the
form:
\begin{gather}
  \max\{ \ip{c}{x} \mid Ax \leq b,~ x \in \{0,1\}^n~ \}, \label{eq:pip}
\end{gather}
where the
columns 
of an  unknown constraint matrix $A \in [0,1]^{d\times n}$ appear one-by-one, and the decisions of the
algorithm are encoded as variables $x_t \in \{0,1\}$.

Although online packing IPs are difficult to solve in the
worst-case, one of the remarkable successes has been for settings
where the PIP instance is chosen adversarially, but the columns are
then presented in a \emph{uniformly random order}. If the
constraints are ``not-too-tight'', we can get  very good
solutions: e.g., for a PIP where the constraint matrix $A \in [0,1]^{d\times n}$
and the entries of $b$ are $\Omega(\e^{-2} \log d)$, we know
algorithms that obtain a $(1-\e)$-approximation 
as
long as the columns of $A$ arrive in a uniformly random
order~\cite{KRTV14,GM-MOR16,AD15}. In fact, these results can
be thought of as extensions of the \emph{multiple-secretary} problem~\cite{Kleinberg-SODA05}
(which is the case where $d=1$ and the matrix $A$ is the all-$1$s
matrix), and ultimately as extensions of the
classical \emph{secretary problem}~\cite{Dynkin-Journal63} where the goal is to pick a single
item to maximize the value.

However, previous algorithms rely heavily on the random-order
assumption, and are susceptible to worst-case corruptions: If an
adversary is allowed to add in a small number of columns that arrive
at specific times, most existing algorithms for the
random-order setting fail disastrously. E.g., the classical
single-item secretary algorithm---which waits for $n/e$ items and then
picks the first item bigger than all preceding items---fails even if
a single high-value item is added at the beginning. Many algorithms for
solving LPs essentially try to estimate duals/thresholds, which can
be skewed by a small number of adversarial items. In short, most current 
algorithms are non-robust, and seem to over-fit to the model.
Our central motivating questions are:
\begin{quote}
  \emph{Can we get robust algorithms for online resource allocation, and
  for PIPs in particular? More broadly, when we give algorithms that
  make assumptions on the data, how do we ensure that their
  performance degrades gracefully as we allow adversarial corruptions?}
\end{quote}

The \emph{Byzantine secretary} model~\cite{KM-ICALP20,BGSZ-ITCS20} is
one attempt to model adversarial corruptions in the random-order
model. In this model, each request (i.e., column of $A$) is either red
(i.e., adversarial) or green (i.e., benign), where these colors are
not known to the algorithm. The adversary chooses the arrival time of
each red column, whereas the green columns choose their arrival times
independently and uniformly at random. The benchmark is now the value
of the optimal green set (sometimes with one green item removed).

In~\cite{BGSZ-ITCS20}, the authors give robust algorithms for the single- and
multi-item secretary settings, but leave open the question of getting
robust algorithms for the setting of packing integer programs. In this
paper, we resolve this question positively, and also give simpler and
better robust algorithms for the single-item case. Our techniques are general and also extend to robust algorithms for PIPs in the
  prophet model, specifically to  the
  Prophet-with-Augmentations model of~\cite{ISW-EC20}.

\subsection{Our Results}

\textbf{PIPs in the Byzantine Secretary Model.}
Our first set of results give robust algorithms for packing integer
programs.  To state the result, let $\Gitems$ denote the green
columns of  matrix $A$, and  let $\gmax$ denote the green column with 
largest value. Let
$\OPT(\Gitems)$ and $\OPT(\Gitems \setminus g_{\max})$ denote the
optimal value of the offline PIP~(\ref{eq:pip}) when restricted to the green
columns, with and without $\gmax$ respectively.

\begin{theorem}[Informal: Robust PIPs]
  \label{thm:inf-main1}
  There exists an algorithm for packing integer programs with
  value  $\Omega(\OPT(\Gitems \setminus g_{\max}))$  in the
  Byzantine secretary model when the entries of $b$ are
  $\Omega(\poly(\log n))$.
\end{theorem}

(See \S\ref{sec:byzant-PIP} for formal statements.)
One may ask whether we can compare to $\OPT(\Gitems)$ instead: sadly,
\cite{BGSZ-ITCS20} showed instances where it is impossible to compare
to $\OPT(\Gitems)$ without further assumptions. However, if we make
(mild) regularity assumptions on the input, we can indeed get much
more nuanced results. Indeed, 
suppose we are given even a very rough estimate of $\OPT(\Gitems)$---to
within \emph{polynomial} in $n$ factors---then we can achieve the following: 

\begin{theorem}[Informal: Robust PIPs with a Rough Estimate]
  \label{thm:inf-main2}
  There exists algorithms for packing integer programs which take a
  rough estimate of the $\OPT(\Gitems)$ and achieve the following
  guarantees:
  \begin{OneLiners}
  \item[(i)] value $\OPT(\Gitems) \cdot \Omega(1)$ for any instance.
  \item[(ii)] value $\OPT(\Gitems)\cdot (1- \alpha - O(\e))$ for
    instances where $\alpha$ is the fraction of red/adversarial
    columns and  the value of any $\approx \log d/\e^4$ columns
    accounts for only an $\eps$-fraction of the optimum (i.e., the
    optimal value is not concentrated on a small set of columns).
  \end{OneLiners}
  Both these algorithms require that the 
  entries of $b$ are at least $B := \Omega(\log d \cdot \log\!\log
  n/\eps^4)$. 
\end{theorem}

\iffull
(See \S\ref{sec:byzant-PIP} and \S\ref{sec:one-plus-eps} for formal
statements.)
\fi
Note that the performance of this algorithm in part\,(ii) approaches $1$ as the
fraction $\alpha$ of adversarial corruptions gets small.
The algorithm of part\,(i) is the same as in \Cref{thm:inf-main1}, but
we need new ideas for part (ii), the case of ``smooth'' instances where
the optimal value is spread out.
\Cref{thm:inf-main1,thm:inf-main2}
both use our technique of \emph{robust threshold estimation}, 
 which is based on 
two conceptually clean ideas
using tools from online learning: Firstly, we show that given an
estimate $\gamma \approx \frac{\OPT(\Gitems)}{B}$, we can pick a set
of items that achieve value close to $\OPT(\Gitems)$. This uses a
low-regret online linear optimization algorithm to learn a good set of
duals, which are then used  to select  items. Secondly, we break the
time horizon into $K$ pieces, and then use another online learning
algorithm to ``learn'' the parameter $\gamma$. This is where we use
our coarse estimate of the optimal value: it allows us to focus on a
set of $O(\e^{-1} \log n)$ possible values for $\gamma$. We discuss
the technical ideas in \S\ref{sec:our-techniques}.

\medskip
\textbf{PIPs in the Prophet-with-Augmentations Model.}  We think that
our approach of using online learning to get robust  algorithms will
be useful in other contexts as well. 
As an example, we consider PIPs in the \emph{prophet model}
where the columns are known up-front, but the value of each column is
independently drawn from a known distribution. In the \emph{prophet-with-augmentations} model of \cite{ISW-EC20}, the adversary is allowed
to  adaptively add arbitrary positive perturbations/augmentations to the
random column values, and the algorithm has to be robust to these
augmentations.
The algorithm competes against the base instance, i.e., the expected value $\OPT_{base}$ of the
offline optimum  when all perturbations are zero. 

It may appear that adding positive perturbations should only improve
the algorithm's performance, but \cite{ISW-EC20} show that the popular
$1/2$-approximation median threshold algorithm~\cite{Samuel-Annals84}
for single-item prophets can become arbitrarily bad due to
augmentations. (The reason is similar to that for the Byzantine
Secretary model: the adversary can present a single high-valued item
in the beginning that is just above the threshold.) \cite{ISW-EC20}
show how to avoid these problems, and  design robust prophet algorithms for the single-item and uniform-matroid problems.
In \S\ref{sec:prophets} we show robust
prophet algorithms for the general setting of PIPs, where our underlying technique is
again based on robust threshold estimation using tools from online learning.

\begin{restatable}[Robust Prophet PIPs]{theorem}{ProphetsAug}
  \label{thm:prophetWithAug} There exists an algorithm for packing integer programs that gets
  value $\Omega(\OPT_{base})$ in the Prophet-with-Augmentations model
  when the right-hand sides of the PIP are
  $\Omega(\log d)$.
\end{restatable}

\textbf{Improved Results for  Single-Item and Matroid Cases.} Our next 
results improves on those of \cite{BGSZ-ITCS20}, for the
case of picking a single item or an independent set in a matroid. (Details in \S\ref{sec:single-item}.)

\begin{restatable}[Single-Item Probability Max]{theorem}{SingleProb}
  \label{thm:main3}
  There is an algorithm for the single-item Byzantine secretary
  problem that picks value at least $\OPT(\Gitems \setminus \gmax)$ with
  probability $\Omega\big(\nicefrac{1}{\log n}\big)$.
\end{restatable}

In the case of a single-item, observe that $\OPT(\Gitems \setminus \gmax)$ is the same as the  value $\val(g_2)$ of the 2nd-highest green item.
The previous result of  \cite{BGSZ-ITCS20} was non-constructive, and only showed \emph{existence} of an algorithm with success probability
$\Omega\big(\nicefrac{1}{(\log n)^2}\big)$; hence our result
improves on the previous results both qualitatively and
quantitatively. We also improve the value maximization results of \cite{BGSZ-ITCS20}.

\begin{restatable}[Value-Maximization for Secretary Problems]{theorem}{ValueMax}
  \label{thm:mainValue}
  There exist algorithms for the following Byzantine secretary
  problems that aim to maximize the expected value of selected items:
  \begin{OneLiners}
  \item[(i)] for the single-item case, we can get expected value
    $\Omega\big(\nicefrac{1}{\log^* n}\big) \cdot \OPT(\Gitems \setminus \gmax)$, and
  \item[(ii)] for the case of a matroid of rank $r$, we can get
    expected value
    $\Omega\big(\nicefrac{1}{\big((\log^* n)^2 \cdot \log r\big)}\big)\cdot
    \OPT(\Gitems \setminus \gmax)$
  \end{OneLiners}
\end{restatable}

\noindent (See  \S\ref{sec:valueMaxi} for proofs.) The former result improves on the the previous expected value of
$\Omega\big(\nicefrac{1}{(\log^* n)^2}\big)\cdot \OPT(\Gitems \setminus \gmax)$, and the latter result improves the previous expected value of 
$\Omega\big(\nicefrac{1}{\log n}\big) \cdot \OPT(\Gitems \setminus \gmax)$ when the rank $r$ is sufficiently smaller than the number of items $n$. 

\subsection{Our Techniques}
\label{sec:our-techniques}

	The general  idea of all our algorithms is to find \emph{robust thresholds}. For packing IPs in the secretary and prophet settings, these robust thresholds are obtained by using a Multiplicative-Weight Updates (MWU) algorithm, and for single-item probability/value maximization the idea is to perform a robust binary search
on a set of candidate thresholds as we gather more information over time. 
Let us now flesh these ideas out in the context of two of our results: for Byzantine PIPs and for
single-item probability maximization.

\medskip
\paragraph{Byzantine Packing IPs.}
 At a high level, our robust algorithm looks at a Lagrangified value
 $c_t x_t - \gamma \ip{\lambda_t}{A_t x_t}$ of the $t^{th}$ item to
 make its decision. Here the dual $\lambda_t$ (computed using MWU)
 puts a relative ``penalty'' on each constraint, with higher penalties for
 constraints that are more occupied. The scale parameter $\gamma$
 balances between the value and the occupation penalty. Given this, our
 algorithm $\alg(\gamma)$ picks item $t$ if its Lagrangian value is
 non-negative, namely if $c_t \ge \gamma \ip{\lambda_t}{A_t}$. 
 Variants of this algorithm have been previously used 
 to study packing IPs in the \emph{stochastic} setting~\cite{AD15,GM-MOR16}. 
 Our first algorithmic contribution is that this algorithm can be made
 robust in the Byzantine Secretary model, \emph{assuming that  the
   right $\gamma$ is known}. Intuitively, the algorithm has a
 ``self-correcting'' nature that balances the occupation of the different constraints, and  the worst-case guarantees of MWU can be used to show its robustness.

 In stochastic settings, the right $\gamma$ is easy to estimate from
 the initial items; however,  a few adversarial items can bias the
 estimation in the Byzantine setting. So our second algorithmic
 contribution is to use a second layer of learning to estimate $\gamma$. We
 break the time horizon into $\parts$ intervals, learn
 $\gamma_1,\gamma_2,\ldots,\gamma_{\parts}$ online and run
 $\alg(\gamma_i, I_i)$ in each interval $I_i$.  The $i^{th}$ reward of
 expert $\gamma$ is the value the algorithm would have obtained on
 interval $I_i$ is on using $\gamma$. Since MWU algorithms choose the
 sequence $(\gamma_i)_i$ to do almost as well as the right $\gamma^*$,
 we get $\sum_i \alg(\gamma_i, I_i) \gtrsim \sum_i \alg(\gamma^*,
 I_i)$, which is essentially the value of the algorithm that knew the
 right $\gamma^*$ in hindsight. One difficulty is that the additive
 regret typically scales with the range of the possible rewards of the
 experts, and the red items can make this range too big. To handle
 this we introduce a truncation to these rewards (note we don't know
 $\OPT(\Gitems)$, so this step needs care), and also use a recent
 multiscale experts result of \cite{BDHN} to make the
 regret scale with the range of the reward of the \emph{best expert},
 not with all the ranges. 
  
 This idea (with some changes) extends to the case of robust
 prophets. Again, we define robust thresholds using MWU: we start off
 considering the Lagrangified value
 $c_t x_t - \gamma \ip{\lambda_t}{A_t x_t}$ of the $t^{th}$ item to
 make our decisions. The challenge of estimating the right $\gamma$
 now becomes simpler since we are given the distributions. But a new
 challenge arises: the expected occupation is different at each time
 step $t$ (which was not the case with random-order). To handle this
 issue and get the eventual solution, we further refine the Lagrangian
 penalty function (see \S\ref{sec:prophets} for details).
	

\medskip
\paragraph{Single-item: probability maximization.} 
The classical secretary algorithm has two phases: \emph{sample} items to estimate a good threshold and \emph{select} an item above the threshold. 
Recall that an adversary can thwart such an algorithm by sending a red
item of very large value in the sample phase (which makes the
threshold too high, so the algorithm does not pick any item). 
However, this failure gives us information: namely, that the max value seen in the sample phase is an upper bound on the value of future items. 
By running $O(\log n)$ copies of the secretary algorithm with distinct
sample phases, we can force the adversary to plant a high-valued red
item in each sample phase, and these values must decrease over time.
(Running these multiple algorithms picks up to $O(\log n)$ items, but
we can subsample down to 1; this is where we lose our approximation factor.)

Now, assuming this nice structure, we can give our \textsc{Search}
procedure that takes an initial set of candidate thresholds and narrows it down as it gets more information, using a robust binary-search procedure. 
The algorithm uses the values of items in an initial prefix $I_0 =
[0,\frac14]$ as a candidate set $\Theta$ for the ``right'' threshold
$\val(g_2)$. At each time, it maintains upper and lower estimates that is used to filter these candidates:
\begin{OneLiners}
\item[1.] The upper bound $\widehat{u}$ is the maximum value seen in
  the previous interval, and the lower bound $\widehat{\ell}$ is the
  maximum value we have picked so far.
\item[2.] The threshold for the current interval is the median of the
  ``surviving'' candidates of $\Theta$, namely those of value in
  $(\widehat{\ell}, \widehat{u}]$. These have value strictly above what we have already picked. 
\item[3.] The algo picks the first item in the interval that is above this threshold. 
\end{OneLiners}

Suppose that, as discussed above, all intervals have a high-value
item; namely, the maximum-value item $\rho_i$ of the $i$th interval
has value $\ge \val(g_2)$. Then the upper bound $\widehat{u}$ never
excludes the ``right'' threshold $\val(g_2)$. Moreover, assume that
these max values $\rho_i$ are non-increasing. In this case, the set of
surviving candidates halves in each interval! Indeed, if we pick an
item in the interval the lower bound increases to the median value,
else all items in the interval were below the median and the upper
bound decreases to the current median value. Since there are more than
$\log n$ intervals, at some point the set of candidates becomes
empty. Now if $g_2$ comes in interval $I_0$ (which happens with
probability $\frac{1}{4}$), the right threshold $\val(g_2)$ is an
initial candidate but the lower bound must have excluded it, so we
have already picked an item with at least this much value.

The idea for our \emph{value-maximization} algorithms is to iteratively
refine the thresholds: we start with a polynomially-approximate
threshold, but each time we pick an item, we prove that we either
get good expected value, or else we get an exponentially-better
threshold. Naturally, this has to be done robustly, so that the red items
have  a limited impact.

\subsection{Further Related Work}
See
\cite{GS-Book20} for general works on random-order online problems. 
In particular, generalizations to matroid and Packing LPs in stochastic models have  been extensively studied, e.g., see~\cite{BIKK-JACM18,Lachish-FOCS14,FSZ-SODA15} for matroids  and \cite{KRTV14,GM-MOR16,AD15} and references therein for packing.
In the last two decades, there is a long line of work  extending 
the classical  single-item $\nicefrac{1}{2}$-approximation prophet inequality~\cite{Krengel-Journal77,Samuel-Annals84} to  packing constraints. In particular, see \cite{KW-STOC12} for matroids and \cite{Rubinstein-STOC16,RS-SODA17} for arbitrary packing constraints. 
For applications of prophet inequalities to pricing mechanisms and online algorithms, we  suggest the  tutorial~\cite{FKS-EC21} and the survey~\cite{Lucier17}.

The above online algorithms for secretary and prophet models rely
heavily on the stochastic assumption, and are susceptible to even
slight worst-case corruptions. We believe that the robust algorithms
in this paper are interesting in their own right. They bridge the gap
between the (optimistic) stochastic and (pessimistic) adversarial
models, which has been a topic of significant interest in both online
algorithms~\cite{Meyerson-FOCS01,MGZ-SODA12,KMZ-STOC15,KKN-STOC15,mixedModelSpikes,Molinaro-SODA17,KM-ICALP20,GKRS-ICALP20,molRobust}
and online learning (see \cite{LykourisML18,guptaCOLT19} and
references within).

The recent paper~\cite{molRobust} considers a Byzantine-type  model with adversarial and stochastic items, and studies Online Convex Optimization and Welfare Maximization problems. A major difference from our work is that both of these problems are unconstrained. In \cite{KM-ICALP20}, the authors consider the Knapsack Secretary problem (i.e., PIP's with a single constraint) in a similar Byzantine model but with the additional assumption that the adversarial items come in bursts. They obtain a $(1-\e)$-approximation when $B \gtrsim \frac{1}{\e^2}$ and there are at most $\approx \frac{n}{\sqrt{B}}$ adversarial items in bursts of size $\approx \sqrt{B}$. 
 Finally, \cite{GKRS-ICALP20} consider streaming problems (i.e. the algorithm has limited memory) in a similar Byzantine model, and design algorithms for max-matching and submodular maximization.




\section{Byzantine Packing Integer Programs}
\label{sec:byzant-PIP}

In this section we discuss how to solve PIPs given in~(\ref{eq:pip})
and get a constant factor of the expected optimal value. The approach
will be to solve the its \emph{linear programming} relaxation: since
we assume that the right-hand sides are large (i.e., $\Omega(\log d)$)
compared to the entries of the constraint matrix, scaling down the
solution slightly and independently rounding each variable immediately
give an integer solution with almost as much value with high
probability. 

Each item $i$ is a pair $(c_i, a_i)$ of a column $a_i \in [0,1]^d$ of
$A$, and its value $c_i \in \R_{\geq 0}$. The $n$ items consist of
$\Gsize$ green items (denoted by $\Gitems$) and $\Rsize = n - \Gsize$
red items (denoted by $\Ritems$). The arrival times of red items are
chosen by an adversary. Then each green item $i \in \Gitems$ chooses
its arrival time independently and uniformly in the time horizon
$[0,1]$.

Our algorithm breaks the time horizon $[0,1]$ into several intervals
and considers the items that fall into each interval seperately. To
argue about this and other objects, it is useful to define the induced
LP (denoted $\textrm{LP}(S)$) for any subset $S$ of columns/items:
\begin{align*}
  \max&~ \ts \sum_{i \in S} c_i x_i  \tag{$\textrm{LP}(S)$}\\
 \text{s.t.}      &~ \ts \sum_{i \in S} a_i x_i \le B \cdot \ones\\
      &~\mathbf{x} \in [0,1]^n \enspace , 
\end{align*}
where $\ones$ is the $d$-dimensional all-ones vector. Let its optimal
solution be $\textbf{x}^*(S)$, having value
$\ip{c}{\textbf{x}^*} = \OPT(S)$. By rescaling rows, we can assume a
common value $B$ on the RHS.
We focus on two benchmarks: $\OPT(\Gitems)$ and
$\OPT(\Gitems \setminus g_{\max})$ where $g_{\max}$ denoting the green
item with the highest value. 

\subsection{Algorithm Outline}
\label{sec:smoothness}

The first step of our algorithm to reduce to solving the following ``smooth'' instances:
\begin{assumption}[Smooth Instance] \label{assum:smooth}
  An instance with $\OPT :=  \OPT(\Gitems)$ is \emph{smooth} if:
  \begin{OneLiners}
  \item[1.] the total value of the green items of value $> \frac{\OPT}{B}$ is at most $\frac{\OPT}{2}$.  
  \item[2.] we are given an estimate $\widehat{O}$ for $\OPT$, such that
    $\widehat{O} \in [\OPT/n, \OPT\cdot n]$.
  \end{OneLiners}
\end{assumption}

\begin{restatable}[Reduction to Smooth Instances]{lemma}{Reduction} \label{lem:smoothRed}
  Suppose $B \ge \Omega(\poly(\log n))$. Given an algorithm to solve
  Packing LPs in the Byzantine Secretary setting that with constant
  probability is $\rho$-competitive w.r.t.\ $\OPT(\Gitems)$ for all
  smooth instances, we can obtain an algorithm which is 
  $\Omega(\rho)$-competitive in expectation for all instances w.r.t.\
  $\OPT(\Gitems \setminus g_{\max})$.
\end{restatable}

Given this reduction (which is proved in \S\ref{sec:appendix-pip}), we
prove our main result for smooth instances:

\begin{theorem}[Algorithm for Smooth Instances]
  \label{thm:combo}
  Suppose $B \ge \Omega\big(\parts \log (d \parts/\delta')\big)$ and $\parts \ge
  \Omega(\log\!\log n)$. The solution returned by
    \Cref{alg:AD-byz} for a smooth Byzantine Secretary instance
 satisfies:
  \begin{OneLiners}
  \item[i.] (Feasibility) The solution always packs into the modified budget of $(B + \parts)\cdot\ones$.
  \item[ii.] (Value) The solution has value at least $\Omega(\OPT)$ with probability at least $1-\delta'$.
  \end{OneLiners}
\end{theorem}

We can scale down each $x_t$ by $\frac{B}{B+K} = 1 - o(1)$ to get a feasible solution
with the same value guarantee up to a constant. Then,
applying \Cref{lem:smoothRed} to~\Cref{thm:combo} gives us the
constant factor approximation of \Cref{thm:inf-main1}.  We now prove
\Cref{thm:combo} in the rest of this section: here are the main
conceptual steps:

\begin{itemize}
\item We break the time horizon into $\parts$ time intervals of equal
  size (we will choose $\parts := \Theta(\log\!\log n)$).  For
  interval $I$, use $\Gitems(I)$ and $\Gsize(I)$ to denote the green
  items and their number in $I$, respectively. Note that
  $\E[\Gsize(I)]=|I| \cdot G$, where $|I|$ is the fraction of total
  time $[0,1]$ covered by interval $I$.

\item In \S\ref{sec:robustIfGamma} we give our algorithm for a single
  interval $I$.  Given a parameter $\gamma$, this algorithm runs a
  low-regret OLO subroutine on a carefully chosen Lagrangification of
  the problem.  Let $\alg(I,\gamma)$ denote the {expected value} that
  the algorithm gets when applied to interval $I$ with parameter
  $\gamma$, where the expectation is over the random arrival times of
  the green items. \Cref{alg:AD-oneint} gives a lower bound on
  $\alg(I,\gamma)$.

\item Finally, in \S\ref{sec:optUnknown} we use a multi-scale low
  regret algorithm to learn the optimal choice of $\gamma$. This
  allows us to combine the single-interval algorithms and prove
  \Cref{thm:combo}.
\end{itemize}


\subsection{Algorithm for a Single Interval}
\label{sec:robustIfGamma}

So in this section we fix an interval $I \sse [0,1]$, and give an algorithm that gets good value in this interval as long as it knows the ``correct'' scalar parameter $\gamma$. We use r.v.s $(C_t, A_t)$ to
denote the value and size of the $t$-th item that appears in this
interval: these depend on which of the items from $\Gitems$ actually
fall into this interval, and on their locations. As mentioned before, the idea of the algorithm is to look at the Lagrangified value $C_t x_t - \gamma \ip{\lambda_t}{A_t x_t}$ to make the decision $x_t \in \{0,1\}$ to pick or not the $t$-th item in the interval. The duals $\lambda_t$'s are computed using an online learning algorithm and put a relative ``price'' on each constraint, with higher prices for
 constraints that are more occupied.

To make this precise, let
$\triangle^{d-1} := \big\{ \mathbf{p} \in [0,1]^d \mid \| \mathbf{p}\|_1 = 1 \big\}$ be the
probability simplex.  Given the algorithm's choice $\algx_t \in [0,1]$ for time step $t$,
define the linear penalty function $f_t: \triangle^{d-1} \to [0,1]$ as
\begin{gather}
  f_t(\lambda) := \ip{\lambda}{A_t \algx_t} \enspace .
\end{gather}
The algorithm for interval $I$ (given a $\gamma$) is then described in \Cref{alg:AD-oneint}.
Note that the algorithm does not depend on the exact arrival times of the items, just on their relative
arrival order and step $t$ denotes
the $t$-th arrival in interval $I$.

\begin{algorithm}
  \caption{IntervalByzLP$(I,\gamma)$}
  \label{alg:AD-oneint}
\begin{algorithmic}[1]  
\For{steps $t = 1, 2, \ldots$ in interval $I$}
\State    Use the low-regret OLO algorithm from \Cref{lem:OLOregret} with $\e = \frac{1}{2}$ on $f_1,\ldots,f_{t-1}$ to get
    $\dual_t \in \triangle^{d-1}$\;
\State    Compute $\algx_t \in [0,1]$ maximizing $x \mapsto C_t x - \gamma \ip{\dual_t}{A_t x}$\;
\State    \textbf{break} if the scaled budget is violated, i.e., if a coordinate
    of $\sum_{s \leq t} A_s
    \algx_s$ exceeds $B\,|I|$
 \EndFor
\end{algorithmic}
\end{algorithm}

In the rest of this section, we prove a lower bound on the value that this algorithm obtains over the interval $I$. For that, let $x^*$ denote the optimal solution consisting only of green items whose value is at most $\frac{OPT}{B}$.  By Assumption \ref{assum:smooth},
this solution has value at least $\frac{OPT}{2}$. Also, define 
for any set $S$ of timesteps 
the Lagrangified
value of the optimal solution using the algorithm's choices of
$\gamma \dual_t$'s as the Lagrangian multipliers:
$
\textstyle  \LOPT(S,\gamma) := \sum_{t \in S} \Big(C_t x^*_t - \gamma
  \ip{\dual_t}{ A_t x^*_t}\Big)  .  
$
Let $\stoch(I) \subseteq \N$ denote
the steps $t$ where the $t$-{th} item in the interval $I$ is
green, which is a random set since each green item chooses its arrival times uniformly at
random in $[0,1]$ whereas the red items choose their arrival times
adversarially. 

	The first step for analyzing our algorithm is showing that it obtains value in $I$ comparable to the Lagrangified value of the optimal solution $x^*$ in this interval.

\begin{lemma}[Value Comparable to Lagrangified $\OPT$]
  \label{lem:AD}
  For any $I \sse [0,1]$ and $\gamma > 0$, 
  we have
  \begin{align}
    \alg(I,\gamma) ~\ge~ \min\Big\{\tfrac{1}{2} \, |I| \gamma B~,~
    \LOPT\big(\stoch(I),\gamma \big) \Big\} - 2\gamma \log d \enspace . \label{eq:AD-bound}
  \end{align}
\end{lemma}

\begin{proof}
Consider the algorithm's run. Let $\tau$ be the number of items seen when the
  algorithm stops; that is, the smallest value $\tau$ such that
  $\sum_{t \le \tau} A_t \algx_t \not \le  |I| \cdot B \cdot \ones$. If the
  algorithm does not exhaust the budget, set $\tau$ to be the number of items
  $|\items(I)|$ in interval $I$. The value of the algorithm is exactly
  $\sum_{t \le \tau} C_t \algx_t$ and its occupation is
  $\sum_{t \le \tau} A_t \algx_t$.
		
    \paragraph{Case 1 (Budget exhausted):} The
  multiplicative-plus-additive guarantees of the low-regret algorithm
  from~\Cref{lem:OLOregret} for $\varepsilon = 1/2$ gives
  \begin{gather*}
\textstyle    \sum_{t \le \tau} f_t(\dual_t) ~\ge~ \frac{1}{2}
    \max_{\lambda \in \Delta^{d-1}} \sum_{t \le \tau} f_t(\lambda) -
    2\log d \enspace.
  \end{gather*}
  Substituting the definition $f_t(\lambda) = \ip{\lambda}{A_t \algx_t}$, we get
  \begin{gather*}
\textstyle    \sum_{t \le \tau} \ip{\dual_t}{A_t \algx_t} ~\ge~
    \frac{1}{2}
    \max_{\lambda \in \Delta^{d-1}} \, \sum_{t \le \tau} \ip{\lambda}{A_t \algx_t} -
    2\log d
   ~ =~     \frac{1}{2} \|A_1 \algx_1 + \ldots + A_\tau
    \algx_{\tau}\|_{\infty} - 2\log d \enspace.
  \end{gather*}
  Hence, we can infer that
  \begin{align}
    \sum_{t \le \tau} C_t \algx_t ~-~ \frac{\gamma}{2} \cdot \|A_1
    \algx_1 + \ldots + A_\tau \algx_{\tau}\|_{\infty}
    &~\ge~ \sum_{t \le \tau} C_t \algx_t - \gamma \sum_{t \le \tau} \ip{\dual_t}{A_t \algx_t} -2\gamma \log d ~\ge~ - \, 2\gamma \log d  ,\label{eq:packIntBase}
  \end{align}
  where the second inequality uses that $\algx_t$ is a
  best-response, and hence is no worse than playing $x = 0$.
  Moreover, if we exhaust our budget, the occupation
  $\| \sum_{t \leq \tau} A_t \algx_t \|_{\infty} \geq |I|\cdot B$, and thus
  \begin{align}
\textstyle    \sum_{t \le \tau} C_t \algx_t ~\ge~ \one{\text{(budget exhausted)}} \cdot \frac{1}{2} \gamma |I| \cdot B - 2\gamma \log d \enspace . \label{eq:AD1}
  \end{align}

  \paragraph{Case 2 (Budget left):} To give a lower bound on the value in the case we do not exhaust our
  budget, we use that for any $t$ the value
  $x^*_t$ is not a better response than 
  $\algx_t$:
  \begin{align*}
    C_t \algx_t ~\ge~ C_t \algx_t - \gamma \ip{\dual_t}{A_t \algx_t} ~\ge~ C_t x^*_t - \gamma \ip{\dual_t}{A_t x^*_t} \enspace.
  \end{align*}
  Summing over all times and using the non-negativity of
  $C_t \algx_t$ to drop the red items,
  \begin{align*}
\textstyle    \sum_{t \le \tau} C_t \algx_t ~\ge~ \sum_{t \in \stoch(I),\, t
    \le \tau} C_t \algx_t ~\ge~ \sum_{t \in \stoch(I),\, t \le \tau}
    \Big(C_t x^*_t - \gamma \ip{\dual_t}{A_t x^*_t}\Big) \enspace . 
  \end{align*}
  But when the algorithm does not exhaust its budget, the RHS is
  precisely $\LOPT\big(\stoch(I),\gamma\big)$, and so
  \begin{gather}
    \sum_{t \le \tau} C_t \algx_t ~\ge~ \one{\text{(budget
        left)}} \cdot \LOPT\big(\stoch(I),\gamma\big) \enspace . \label{eq:AD2}
  \end{gather}
  Combining \eqref{eq:AD1} and (\ref{eq:AD2}) concludes the proof.
\end{proof}

	 \IGNORE{
	\begin{remark}
	The main advantage of using this modified version of Agrawal-Devanur is that since $f_t$'s are always non-negative we can just drop them when doing 
	\begin{align*}
		\sum_{t \le \tau} \bigg(c_t \algx_t - \gamma f_t(\algx_t)\bigg) \red{\ge} \sum_{t \le \tau, t \in \stoch} \bigg(c_t \algx_t - \gamma f_t(\algx_t)\bigg) \ge \sum_{t \le \tau, t \in \stoch} \bigg(c_t x^*_t - \gamma f_t(x^*_t)\bigg).  
	\end{align*}
	This is important because we can only properly control $f_t(x^*_t)$ for the stochastic times. 
	\end{remark}
}

	The final piece is to lower bound the Lagrangified value of the optimal solution $x^*$ on the interval $I$. The proof of this lemma crucially uses the random arrival times of the green items.

\begin{lemma} \label{lemma:LOPT}
  Let $|I| \leq \tfrac14$ and $B \ge \Omega\big(\frac{\log (4d/\delta)}{|I|}\big)$. 
  Then for any $\gamma > 0$, 
  \begin{align*}
    \Pr\Big[ \LOPT\big(\stoch(I),\gamma\big) ~\ge~ |I| \cdot \big( \tfrac{\OPT}{4} - 4\gamma B
    \big) \Big] \geq 1-\delta \enspace ,
  \end{align*}
  where the probability is taken over the random arrival times of the green items.
\end{lemma}

We give the essential intuition  of this lemma
 here (at least in expectation)  and defer the details to \Cref{sec:lemma:LOPT}.
Consider any timestep $t \in \stoch(I)$:
  \begin{align}
    \E \Big[C_t x^*_t - \gamma \ip{\dual_t}{ A_t x^*_t}\Big]
    ~\ge ~ \tfrac{\OPT}{2\Gsize} - \gamma \E\Big[\ip{\dual_t}{A_t
      x^*_t}\Big] \enspace , \label{eq:start-of-Llemma}
  \end{align}
  where the expectation is over the random ordering. For intuition
  only, suppose each column $A_t$ is an i.i.d.\ sample (this is not
  w.l.o.g.), so that $A_t x^*_t$ is independent of $\dual_t$.  Then
  the expectation can be pushed into the inner product; hence if there
  are $\Gsize$ green items overall, the expected value
  $\E[A_t x^*_t] \leq B/\Gsize \cdot \ones$, and so
  \begin{align}
    \E \Big[C_t x^*_t - \gamma \ip{\dual_t}{ A_t x^*_t}\Big]
    ~~\ge~~ \tfrac{\OPT}{2\Gsize} - \gamma \ip{\E \dual_t}{\tfrac{ B}{\Gsize} \cdot \ones} ~~\ge~~ \tfrac{\OPT}{2\Gsize} - \gamma \tfrac{B}{\Gsize} \enspace , \label{eq:midLlemma}
\end{align}
where the last inequality uses that
$\dual_t \in \Delta^{d-1}$. Finally, 
$\E[|\stoch(I)|] = |I| \cdot \Gsize$, so
we get $\E [ \LOPT] \ge |I|(\frac{\OPT}{2} - \gamma B)$ to complete the
proof of \Cref{lemma:LOPT} in expectation. 
However, 
the reason why this is just intuition and not a proof is that we sample \emph{without replacement},
so $\dual_t$ (which depends on the $t-1$ first items on the
interval $I$) is correlated with $A_t x^*_t$.
To handle this, in \Cref{sec:lemma:LOPT} we have to argue 
why these correlations are small.

Combining these lemmas gives the desired guarantee for our algorithm.

\begin{lemma}[Value of \Cref{alg:AD-oneint}] \label{lemma:packInterval}
	Let $|I| \leq \tfrac14$ and $B \ge \Omega\big(\frac{\log (4d/\delta)}{|I|}\big)$. Then for any $\gamma > 0$, with probability at least $1-\delta$ we have 
	\begin{align*}
	\alg(I, \gamma) \ge |I| \cdot \min\Big\{\tfrac{\gamma B}{2}\,,\, \tfrac{\OPT}{4} - 4 \gamma B \Big\} - 2 \gamma \log d.
	\end{align*}
\end{lemma}

As mentioned earlier, we see that a good choice of $\gamma$ is $\Theta(\frac{\OPT}{B})$, but this requires us to know the
value of $\OPT$; we now show how another layer of online learning can learn this value well enough.


\subsection{A Robust LP Algorithm via Learning the Multiplier $\gamma$}
\label{sec:optUnknown}

We partition the time interval $[0,1]$ into $\parts$ intervals
$I_1, I_2, \ldots, I_{\parts}$ of equal size, and
run~\Cref{alg:AD-oneint} in each interval with the value of $\gamma$
learned from previous intervals via an online learning
algorithm. Formally, define the gain functions:
\begin{align*}
	\alg_i(\gamma) &:= \textrm{value from
                \Cref{alg:AD-oneint} with parameter $\gamma$ over the $i^{th}$ interval $I_i$ } \\
                  \bar{\alg}_i(\gamma) &:= \min\{ \alg_i(\gamma)~, ~\tfrac{B}{\parts}  \gamma\} \text{  the algorithm's  truncated value} \enspace .
\end{align*} 
Using Assumption \ref{assum:smooth} that we know $\OPT$ up to $\poly(n)$ factors, let
$\Gamma$ be a list of $O(\log n)$ values that contains
$\frac{\OPT}{16B}$ within a factor of 2. The complete algorithm is the
following:

\begin{algorithm}
  \caption{ByzLP}
  \label{alg:AD-byz}
\begin{algorithmic}[1]
  \For{interval $i = 1, \ldots, K$}
    \State
    pick ${\gamma}_i \in \Gamma$ using the multiscale experts
    algorithm of \Cref{thm:multiscale} on
    $\bar{\alg}_1,
    \bar{\alg}_2,\ldots, \bar{\alg}_{i-1}$.
    \State
    run \Cref{alg:AD-oneint} over interval $I_i$
    with parameter ${\gamma}_i$, thereby getting
    value $\alg_i({\gamma}_i)$.
    \EndFor
\end{algorithmic}
\end{algorithm}

\begin{proof}[Proof of \Cref{thm:combo}]
  Since there are $\parts$ intervals, the feasibility follows directly from the stopping rule for the algorithm, and the fact that item sizes are at most $1$. For the second claim, the
  value of the algorithm is $\sum_i \alg_i({\gamma}_i)$, which is at
  least $\sum_i \bar{\alg}_i({\gamma}_i)$, so it suffices to lower
  bound the latter. Let $\gamma^*$ be a value in $\Gamma$
  that is in $\left[\frac{1}{32}\frac{\OPT}{B}, \frac{1}{16}\frac{\OPT}{B}\right]$. 
  Due to the truncation $\bar{\alg}_i(\gamma^*)$ is $O(\frac{\OPT}{\parts})$, so
  the multiscale regret guarantee of \Cref{thm:multiscale} 
  with $\e = \Theta(\sqrt{\log |\Gamma| / K})$ gives that
  in every scenario, 
  \begin{align}
    \sum_i \bar{\alg}_i({\gamma}_i) ~\ge~ \sum_i \bar{\alg}_i(\gamma^*) - O(1) \cdot \sqrt{\parts \log |\Gamma|} \cdot  \tfrac{\OPT}{\parts}    ~=~\sum_i \bar{\alg}_i(\gamma^*) - \OPT \cdot  O(\tfrac{\sqrt{\log \log n}}{\sqrt{\parts}}) \enspace . \label{eq:packingWrap1}
  \end{align}
  Using the guarantee of \Cref{lemma:packInterval} with $\delta = \frac{\delta'}{\parts}$, we have that with probability at least $1 - \frac{\delta'}{\parts}$, (using
  $\alg(I,\gamma^*)$ for the value obtained by the algorithm of the
  previous section, and $I_i$ for the $i$-th interval)
  \begin{align*}
    \bar{\alg}_i(\gamma^*) 
      ~\ge~ \min\big\{ \alg(I_i, \gamma^*)~,~\tfrac{\gamma^* B}{\parts}\big\} ~\ge~ \tfrac{1}{\parts}\cdot \min\big\{\tfrac{\gamma^* B}{2}~,~\tfrac{\OPT}{4} - 4\gamma^* B\big\} - 2\gamma^* \log d \enspace ,
        \end{align*}
      where the last inequality uses \Cref{lem:AD} and \Cref{lemma:LOPT}. Since  $\gamma^*\in \left[\frac{1}{32}\frac{\OPT}{B}, \frac{1}{16}\frac{\OPT}{B}\right]$, we  get
      \[    \bar{\alg}_i(\gamma^*)     
      ~\ge~ \left(\tfrac{1}{\parts}\cdot \min\big\{\tfrac{1}{64}, \tfrac{1}{4} - \tfrac{1}{8} \big\} - \tfrac{\log d}{8B}\right)\OPT
      ~\ge~\Omega(1)\cdot \tfrac{\OPT}{\parts}\enspace .
      \]
 Taking sum over all $\parts$
  intervals we get with probability
  at least $1-\delta'$ that
  $\sum_i \bar{\alg}_i(\gamma^*) \ge \Omega(\OPT).$
  Using this on \eqref{eq:packingWrap1} and
  using $\parts \ge \Omega(\log \log n)$ concludes the
  proof of \Cref{thm:combo}.
\end{proof} 

In
\iffull
\S\ref{sec:one-plus-eps}
\fi we give an algorithm that gets an approximation
approaching $1$ when the number of red items gets small, thereby proving
\Cref{thm:inf-main2}.

%
%




\newcommand{\good}{nice}
\newcommand{\GoodInt}{N}
\newcommand{\Good}{\widehat{\mathcal{N}}}
\newcommand{\GOOD}{\mathcal{N}}
\newcommand{\High}{\widehat{\mathcal{H}}}
\newcommand{\HIGH}{\mathcal{H}}
\newcommand{\SearchII}{\textsc{SearchII}\xspace}
\newcommand{\Sample}{\textsc{Sample}\xspace}

\newcommand{\target}{C^*}
\newcommand{\gSecMax}{g_2}
\newcommand{\arrival}{\textsf{time}}
\newcommand{\lowBd}{L}
\newcommand{\upBd}{U}

\section{Byzantine Secretary for Single-Item Probability Maximization}
\label{sec:single-item}

We now consider the (single-item) Byzantine Secretary Problem, where the online
model is exactly the same as in the previous section but now we can
only pick one item: the goal is to maximize the \emph{probability} of selecting an item of value at least $\OPT(\Gitems\setminus\gmax)$,
i.e., the value of the second-most valuable green item. We show the
following:

	\SingleProb*
		
This improves on the algorithm of \cite{BGSZ-ITCS20}, which 
(a)~is nonconstructive, relying on the use of Yao's minimax principle, and (b) succeeds with a smaller probability of
$\Theta(1/\log^2 n)$.  

\subsection{Algorithm}
\label{sec:prob-max}

Our algorithm is based on two procedures, \textsc{Struct} and
\textsc{Search}. Both pick $K = O(\log n)$ items. We will show that
with constant probability at
least one of \Struct and \Search succeeds in picking an item with
value at least $\OPT(\Gitems\setminus\gmax)$. 
By picking all items chosen by either procedure, we get
an algorithm that picks $2K = O(\log n)$ items and succeeds with
probability $\Omega(1)$. Now picking uniformly at random one of these $2K$ items,
we get an algorithm that picks a single item of value at least $\OPT(\Gitems\setminus\gmax)$
with probability $\Omega(1/\log n)$. 
In \S\ref{sec:our-techniques} we outline the intuition
behind these algorithms.

For both our procedures, we partition the time interval $(\frac14, \frac34]$ into
$K=\Theta(\log n)$ intervals $I_1,\dots, I_K$, each of equal width
$\frac{1}{2K}$. Let $I_0 := [0,\frac14]$. For all $i$, let $\rho_i$ be
the maximum value of a \emph{red item} in interval $I_i$ and let
$\widehat{\mu}_i$ be the maximum value of any item in interval $I_i$
(note that $\rho_i$ is deterministic but unknown to the algorithm and $\widehat{\mu}_i$ depends on
when the green items arrive). For an item $e$, let $\arrival(e)$ and $\val(e)$
denote the arrival time and value of $e$ respectively. Recall that $\gSecMax$ is the 2nd-most valuable green item. To simplify the notation, let $C^* := \OPT(\Gitems\setminus\gmax) = \val(\gSecMax)$.

The procedure \Struct runs $K$ independent subroutines similar to the classic single-item secretary algorithm: the $i^{th}$ one
sets a threshold  $\widehat{\mu}_i$ and picks the first item after
 $I_i$ that reaches this threshold.

\begin{algorithm}
  \caption{Procedure \Struct}
  \label{procedure:struct}
\begin{algorithmic}[1]  
  \For{value $i = 1, 2, \ldots, K$ in parallel}
    \State$\widehat{\mu}_{i} := $ maximum value of any
    item seen in interval $I_{i}$
    \State pick the first item (if any) arriving after interval $I_i$ with value at least $\widehat{\mu}_i$. 
  \EndFor
\end{algorithmic}
\end{algorithm}

The second procedure \Search maintains a \emph{set of candidate thresholds}
containing a subset of all the values $\widehat{\Theta}_1$ seen in
interval $I_0 = [0,\tfrac14]$. In each interval, the current candidate
set is obtained by focusing on these values lying between the maximum
value of an item seen in the previous interval, and the largest
value item picked by this procedure so far. The threshold for the current interval is set to the median of
these values. (We define $\median(\emptyset) = -\infty$.)

\begin{algorithm}
  \caption{Procedure \textsc{Search}}
  \label{alg:search}
\begin{algorithmic}[1]
  \For{interval $i = 1, 2, \ldots, K$}
    \State $\widehat\lowBd_i \gets $ max value of an item already picked by this procedure, $-\infty$ if no item has been picked
    \State $\widehat{\upBd}_i \gets \widehat{\mu}_{i-1}$ := maximum value of any item seen in interval $I_{i-1}$ 
    \State $\widehat{\Theta}_i \gets \{\val(e) \mid \arrival(e)\in [0, \frac14] \text{
      and } \widehat \lowBd_i < \val(e) \le \widehat \upBd_i\}$ 
    \State pick the first item (if any) in the interval $I_i$ having value at least
    $\median(\widehat{\Theta}_i)$ 
  \EndFor
\end{algorithmic}
\end{algorithm}

\subsection{Analysis}
It is clear that both procedures pick at most $K = O(\log n)$ items.
We show that with probability $\Omega(1)$, at least one of the two
procedures picks an item with value at least $\target$.  The intuition is
this: if the maximum value items in each interval are monotone
decreasing and greater than $\target$, then in each interval procedure \Search
halves the candidate set size. If this set contains the value
$\target$ (which happens, e.g., when $g_2$ arrives in the interval $I_0$),
then we must eventually pick an item of large value. Of course, the
maximum value items may not be monotone, and they may have values
below $\target$, but then we show that \Struct gets large value.

We now give the proof details.  Recall that $\rho_i$ is the maximum
value of any red item in interval $I_i$; define $\rho_i \gets -\infty$
if no such items exist. We say an interval $I_i$ (including $I_0$) is \emph{high} 
if $\rho_i \ge \target$, and \emph{low} otherwise. Recall that the property of being high
just depends on the location of the red items, which we  assume are
deterministically placed.

\begin{lemma}\label{lem:alg1-good}
For any value of $K$, \Struct succeeds (i.e. picks an item of value at least $\target$) with probability $\Omega(1)$ if either of the following properties fail:
\begin{OneLiners}
\item[(i)] If $I_i$ and $I_j$ are high intervals with $i<j$, then
  $\rho_i > \rho_j$.
  
\item[(ii)] There are fewer than $\frac{K}{4}$ low intervals.
\end{OneLiners}
\end{lemma}

\begin{proof}
  Suppose property (i) fails, and let $I_i$ and $I_j$ be high intervals
  with $i<j$ and $\rho_i \le \rho_j$. The highness of $I_i$ implies
  $\rho_i \ge \target$. Now if $\gmax$ does not fall in $I_i$ (which
  happens with probability $1 - \nicefrac{1}{2K}$), then the maximum
  value $\widehat{\mu}_i$ in this interval is $\rho_i$, and therefore
  the run of \Struct corresponding to interval $i$ sets a threshold of
  $\rho_i \geq \target$.  This run would definitely select the item
  corresponding to $\rho_j$ in $I_j$, if it has not picked an item
  earlier; this selected item has value at least $\rho_i \geq \target$.

  Else suppose property (ii) fails, and there are at least $\nf{K}4$ low
  intervals. If $g_2$ arrives in a low interval $I_i$, then we set the
  threshold to be $\val(g_2) = \target$; now if $\gmax$ arrives in $(\frac34, 1]$, the
  $i^{th}$ run of \Struct is guaranteed to pick an item. Since there
  are at least $\frac{K}{4}$ low intervals, this event occurs with
  probability at least
  $\frac{K}{4}\cdot \frac{1}{2K} \cdot \frac14 \ge \frac{1}{32}$.
\end{proof}

The remainder of the proof shows that if the two properties of
\Cref{lem:alg1-good} are indeed satisfied, then \Search succeeds with
constant probability. Define
$\upBd_i := \min\{\rho_j \mid j < i \text{ and } I_j$ is high$\}$, and
\[ \textstyle \Theta_i := \lbrace \val(e) \mid \arrival(e)\in
  \left[0,\frac14\right] \text{ and } \widehat{\lowBd}_i < \val(e) \le
  \upBd_i\rbrace. \] The only difference from $\widehat{\Theta}_i$ is that
the upper bound is $\upBd_i$ instead of $\widehat{\upBd}_i$.  (The intuition
is that if property~(i) of \Cref{lem:alg1-good} holds, then loosely
speaking $\widehat{\Theta}_i$ and $\Theta_i$ should behave similarly,
and we can argue about the latter instead of the former.) Observe that $\widehat{\lowBd}_i$ can only increase and $\upBd_i$ can only decrease, so $\Theta_i\supseteq \Theta_{i+1}$.

To make this precise, define an interval $I_i$ (for $i\ge 1$) to be \emph{\good} if
the intervals $I_i$ and $I_{i-1}$ are both high, and moreover the item
$\gmax$ does not arrive in interval $I_{i-1}$. 
This property of being {\good} does depend on the location of the top green item 
$\gmax$, but otherwise is independent of the random locations of other green 
items. 

\begin{lemma}\label{fct:halving}
  If property~(i) from \Cref{lem:alg1-good} holds and the interval $I_i$ is {\good}, then $|\Theta_{i+1}| \le \frac12 |\Theta_i|$.
\end{lemma}
\begin{proof}
  Property~(i) of \Cref{lem:alg1-good} means the $\rho$-values of the
  high intervals are in decreasing order, and the minimum in the
  definition of $\upBd_i$
  is achieved at the
  last high interval before $i$.  Since interval $I_{i-1}$ is high, we
  have $\upBd_i = \rho_{i-1}$. Moreover, $\gmax$ does not arrive in interval
  $I_{i-1}$, and therefore $\rho_{i-1} = \widehat{\mu}_{i-1}$. This in
  turn means that $\upBd_i = \widehat{\mu}_{i-1} =  \widehat{\upBd}_i$, and 
  $\widehat{\Theta}_i = \Theta_i$. 

  If $m_i := \median(\Theta_i) = \median(\widehat{\Theta}_i)$, then
  define $\Theta_i^+ := \{v\in \Theta_i \mid v > m_i\}$ and
  $\Theta_i^- := \{v\in \Theta_i \mid v < m_i\}$. Each has
  size at most $\frac12 |\Theta_i|$. If the \Search procedure chose an
  item in $I_i$, then $\widehat{\lowBd}_{i+1} \geq m_i$ and therefore
  $\Theta_{i+1} \sse \Theta_i^+$. Otherwise \Search did not chose an
  item in $I_i$, so all its items must have been smaller than $m_i$;
  in particular, $\rho_i < m_i$. Since $I_i$ is high, $u_{i+1} \le \rho_i < m_i$
  and $\Theta_{i+1}\sse \Theta_i^-$.
\end{proof}

\begin{lemma} \label{lemma:singleProbMaxMain}
  Let $K = 2\log_2 n+4$. Then with probability $\Omega(1)$, at least
  one of \Struct and \Search picks an item of value at least $\target$.
\end{lemma}
\begin{proof}
  By Lemma~\ref{lem:alg1-good}, if either of its properties~(i)
  or~(ii) fails then \Struct succeeds with constant probability. So
  suppose both properties hold. Now condition on the location of
  item $\gmax$; this decides on the {\good}ness of the
  intervals. Property (ii) being satisfied means there are at least
  $K- \frac{K}{4}-\frac{K}{4} - 1 = K/2 - 1 > \log_2 n$ {\good} intervals. 
  (Indeed, we may discard at most $K/4$ intervals because
  $I_{i-1}$ is bad, $K/4$ others because $I_{i}$ is bad, and one more
  because $\gmax$ falls in $I_{i-1}$.) 

  Now let us condition on the event that $g_2$ arrives in the time
  interval $[0,\frac14]$, which happens with probability $\nf14$.
  Since $|\Theta_1| \le n$,
  applying \Cref{fct:halving} (which relies on property~(i)) to each
  of the {\good} intervals implies that after $\log_2 n$ {\good}
  intervals, we get to some index $i$ for which $\Theta_i$ is empty,
  and interval $I_i$ is {\good}. The upper bound $\upBd_i$ for
  $\Theta_i$ is at least $\target$, by definition.
  Since
  item $g_2$ arrived in $[0,\frac14]$ but yet
  $\target = \val(g_2) \notin \Theta_i$, the only reason would be that
  $\widehat{\lowBd}_i \geq \target$. This means that \Search must have already
  chosen an item of value at least $\target$.
\end{proof}

\begin{proof}[Proof of \Cref{thm:main3}]
  Since \Struct and \Search pick at most $K$ items each, the final
  algorithm is to run a random one of these two algorithms, and to
  randomly output one of the $K = O(\log n)$ items picked by that
  algorithm. This gives an item of value at least $\target$ with
  probability $\Omega(\nicefrac{1}{\log n})$.
\end{proof}


\section{Prophet-with-Augmentations for Packing Integer Programs}
\label{sec:prophets}

To show the power of our robust threshold selection idea from
\S\ref{sec:byzant-PIP}, we use it to give
robust algorithms for PIPs in the \emph{prophet model} as well. Recall
that in the classic prophets model, the inputs are drawn from
independent (but possibly non-identical) distributions. Here we
consider the \emph{Prophets-with-Augmentations} model~\cite{ISW-EC20}, in which an adversary is allowed to perturb the
values by adding non-negative values. We now show how to make
PIP algorithms robust to such
perturbations. 

\subsection{Model and Notation}

We are given a \emph{base prophet} instance
$((a_1, \cD_1), \ldots, \allowbreak (a_n, \cD_n), B)$  with budget
$B\cdot \ones$ and $n$ items whose values will be perturbed by an
adversary.  The $t^{th}$ item has a known deterministic size vector
$a_t \in [0,1]^d$ and an initially unknown value $V_t \geq 0$ that is
drawn independently from a known distribution $\cD_t$.  Items values are revealed one-by-one, and before the
$t$-th item's value is revealed, an adversary adds an unknown
perturbation $R_t \geq 0$ to $V_t$. This perturbation may depend on
the history up to (and including) time $t$, as well as the algorithm's
decisions up to time $t-1$. The player then sees $C_t := V_t + R_t$
and has to immediately pick or reject the item. The goal is to pick a
set of items that pack into the known budget $B \cdot \ones$, in order to
maximize the sum of seen values $C_t$ of the picked items.

The algorithm competes against the base instance, i.e., the expected
offline optimum when all perturbations $R_t$ are zero. Let
$\OPT_{base}$ be the value of this expected offline optimum. The main
result of this section is as follows.

\ProphetsAug*

\IGNORE{
{\color{red}
  \begin{question}
    The theorem might still work if we also allow the adversary to
    subtract a value $S_t$ from the item, as long as $S_t \ge
    \frac{1}{2} V_t$, say? 
  \end{question}
  }
  {\color{blue}Not sure how much we want to discuss this.}
}
	
To make our proofs simpler, we assume w.l.o.g.\ that there are no
``large values''. Here we sketch the proof; the full proof is deferred to \S\ref{sec:redSmoothProphet}.
\begin{assumption}
  \label{assum:smoothProphet}
  Each distribution $\cD_t$ is supported on values that are at most $\frac{\OPT_{base}}{20}$. 
\end{assumption}

\begin{proof}[Proof Sketch] 
	Consider running simultaneously an algorithm that obtains a constant approximation under Assumption \ref{assum:smoothProphet} and also an algorithm that picks the first item that takes a value above $\frac{\OPT_{base}}{20}$. Intuitively, in the scenarios where all items come up with value at most $\frac{\OPT_{base}}{20}$ the approximation guarantee of the first algorithm kicks in, and in the remaining scenarios the second algorithm already guarantees value at least $\frac{\OPT_{base}}{20}$; overall, we should get a constant approximation. Running both algorithm may lead to budget occupation up to $(B + 1) \cdot \ones$, but rescaling the solution by $\frac{B}{B+1} > 1-o(1)$ restores feasible while maintaining the same approximation guarantee. 
\end{proof}

\subsection{Algorithm}

The idea of the algorithm is similar to that used for Byzantine PIPs
in \S\ref{sec:byzant-PIP}: to make decisions $x_t \in \{0,1\}$ based
on the Lagrangified value $C_t x_t - \gamma \ip{\lambda_t}{a_t x_t}$,
for a scaling factor $\gamma$. As before, the ``right'' value for
$\gamma$ is $\approx \frac{\OPT_{base}}{B}$. Previously we learned
$\gamma$ over multiple intervals using online learning (since $\OPT$
was not known), we can now directly compute it using the known value
distributions $\cD_t$'s. However, new challenges arise. Firstly, we
will again need to bound a Lagrangified value of the offline optimum
with \emph{good probability} (as in \Cref{lemma:LOPT}), but since the
optimal solution's decisions to pick items depend on the outcomes of
the values of \emph{all} other items, we are not guaranteed to have
any concentration. (Previously $\OPT$ was such that picking the $t$-th
item only depended on the identity of that item.) To fix this, we
compare not against the optimal solution, but against a surrogate
$\psi_t(V_t) \in \{0,1\}$ that makes decisions about item $t$ based
only on its base value $V_t$. Another challenge is that the expected
occupation of such a solution is different in each time step $t$
(which was not the case in random-order, see \Cref{eq:midLlemma}):
this makes it harder to bound the quantity
$\ip{\lambda_t}{a_t\, \psi(V_t)}$. To fix this, we define the
Lagrangian in terms of the modified penalty function
$f_t(\lambda) := \ip{\lambda}{a_t x_t - a_t\, \E \psi_t(V_t)}$.  On a
technical note, since this penalty can be negative, we consider
$\lambda$ values in the ``full-dimensional simplex'' (including the 0
vector), namely
$\blacktriangle^d := \{ \lambda \in [0,1]^d : \sum_i \lambda_i \le
1\}$.

	To make this precise, we start with the existence of the good solution $\psi_1(V_1),\ldots,\psi_n(V_n)$ for the base prophet instance. Similar solutions algorithms have been previously designed for related
problems (e.g., see~\cite{Alaei-SICOMP14,AHL-EC12}), and we defer the
proof to Appendix~\ref{sec:non-robust-prophet}.
	
\begin{lemma} \label{lem:prophBoundOPT} Given a base prophet instance
  $((a_1, \cD_1), \ldots, \allowbreak (a_n, \cD_n), B)$, there are 
  functions $\psi_1, \ldots, \psi_n$, where each $\psi_t$ maps the
  value $V_t$ to a decision in $\{0,1\}$ 
  such that:
  \begin{enumerate}
  \item Total expected value~ $\E [\sum_t V_t\, \psi_t(V_t)] \ge \frac{\OPT_{base}}{4}$. \label{eq:relaxBndOPT}
  \item Total expected utilization~ $\E [\sum_t a_t \,\psi_t(V_t)] \le  \frac{B}{4}\cdot \ones$. \label{eq:budgetOPTbase}
  \end{enumerate} 
\end{lemma}

We now describe our robust algorithm for the prophet-with-augmentations
model. To simplify notation, define $x^*_t := \E[\psi_t(V_t)]$, that is, the probability that this solution picks the $t$-th item. 
Our algorithm requires these $x^*_t$'s, but since they only depend on the distributions of the $V_t$'s thay can be computed a priori. It also needs to know $\OPT_{base}$ to set the value of $\gamma$, which can also be computed a priori (and as the proof shows, a constant-factor approximation to this value suffices).

\begin{algorithm}
  \caption{Procedure \textsc{Prophet-with-Augmentations}}
  \label{alg:AD-prophet}
  \begin{algorithmic}[1]
  \For{steps $t = 1, 2, \cdots$}
    \State 
    Compute $\lambda_t \in \blacktriangle^d$ by using the
    low-regret algorithm of \Cref{lem:OLOregret} with $\e = \frac{1}{2}$ on the functions
    $f_1,\ldots,f_{t-1}$, where $f_t(\lambda) := \ip{\lambda}{a_t x_t - a_t x^*_t}$.
    \State
    Compute $x_t \in [0,1]$ maximizing $x_t \mapsto C_t x_t - \gamma
    \ip{\lambda_t}{a_t x_t - a_t x^*_t}$, where $\gamma =
    \frac{\OPT_{base}}{B}$.  
    \State
    \textbf{break} if the budget is exhausted, i.e., if a coordinate
    of $\sum_{s \leq t} a_s
    \algx_s$ exceeds $B$. Let $\tau$ denote this stopping time step, where $\tau=n$ if the budget is never exhausted.
  \EndFor
\end{algorithmic}
\end{algorithm}


\subsection{Analysis} 

	We now analyze the above algorithm, proving that it attains the guarantee stated in \Cref{thm:prophetWithAug} under Assumption \ref{assum:smoothProphet} (without loss of generality). First, the algorithm violates the budget by at most +1, but this is
easily fixed by rescaling or subsampling, so we ignore this issue
henceforth. To prove that it gets high value, fix a scenario of
$V_t$ and $R_t$. In this scenario, since our decision $x_t$ is a best
response, it is at least as good as $\psi_t(V_t)$. In other words, the
best-response ensures
\begin{align*}
  &\underbrace{\sum_{t \le \tau} C_t x_t}_{= \alg} - \gamma
    \underbrace{\sum_{t \le \tau} \ip{\lambda_t}{a_t(x_t -
    x^*_t)}}_{=: M_{lhs}} ~\ge~ \sum_{t \le \tau} C_t \; \psi_t(V_t) -
    \gamma \underbrace{\sum_{t \le \tau}
    \ip{\lambda_t}{a_t(\psi_t(V_t) - x^*_t)}}_{=: M_{rhs}}. 
\end{align*}
Rewriting, we get
\begin{align} \label{eq:algProphLower}
  \alg ~\ge ~ \sum_{t \le \tau} C_t \;\psi_t(V_t) + \gamma M_{lhs} - \gamma M_{rhs}.
\end{align}
To lower bound the value of $\alg$, we first lower bound $M_{lhs}$. From
the regret guarantee in
\Cref{lem:OLOregret}, we can compare against the action $\lambda=0$ 
to infer
\begin{align}
  M_{lhs} &~\ge~ - O(\log d). \label{eq:lhs1}
\end{align}
Similarly, comparing against the action $\lambda = e_i$ we get (using $|\ip{e_i}{a_t(x_t -x^*_t)}| \le \ip{e_i}{a_t x_t} + \ip{e_i}{a_t x^*_t}$) 
\begin{align*}
  M_{lhs} ~\ge~ \frac{1}{2} \ip{e_i}{\tsum_{t \le \tau} a_t x_t} - \frac{3}{2} \ip{e_i}{\tsum_{t \le \tau} a_t x^*_t} - O(\log d),
\end{align*}
	which then implies 
\begin{align}
  M_{lhs} &~\ge~ \frac{1}{2} \bigg\|\sum_{t \le \tau} a_t x_t
            \bigg\|_{\infty} - \frac{3}{2}\bigg\|\sum_{t \le \tau} a_t
            x^*_t \bigg\|_{\infty} - O(\log d) \notag \\ 
          &~\ge~ \frac{1}{2} \bigg\|\sum_{t \le \tau} a_t x_t
            \bigg\|_{\infty} - ~\frac{3B}{8}~ - O(\log
            d), \label{eq:lhs2} 
\end{align}
where the last inequality uses that $\|\sum_{t \le \tau} a_t x^*_t
\|_{\infty} \leq B/4$ due to \Cref{eq:budgetOPTbase} in \Cref{lem:prophBoundOPT} and the definition $x^*_t = \E[\psi_t(V_t)]$.

For the scenario when $\tau = n$,   \eqref{eq:algProphLower} and
\eqref{eq:lhs1} give that  
\begin{align*}
  \alg ~\ge~ \sum_\blue{t} C_t \; \psi_t(V_t) - \gamma \cdot O(\log d)  - \gamma M_{rhs} .
\end{align*}
Note that the sum is now over all $t$, not only $t\le \tau$.
For the other scenario, when $\tau < n$ (i.e., we exhaust the
budget), we know that $\|\sum_{t \le \tau} a_t x_t\|_{\infty} > B$,
and so using \eqref{eq:algProphLower} and \eqref{eq:lhs2} we get
\begin{align*}
  \alg ~\ge~ \frac{\gamma B}{8}  - \gamma \cdot O(\log d)  - \gamma M_{rhs}.
\end{align*}	
Taking a minimum of both the scenarios $\tau=n$ and $\tau<n$, we get
the following bound that holds for every scenario:
\begin{align*}
  \alg ~\ge~ \min\bigg\{\sum_t C_t \; \psi_t(V_t)~,~  \frac{\gamma B}{8}
  \bigg\} ~- \gamma \cdot O(\log d)  ~- \gamma M_{rhs}. 
\end{align*}
To calculate $\E[M_{rhs}]$, notice that both $\one(\tau \ge t)$ and
$\lambda_t$ depend only on the history up to time $t-1$. (This is where we crucially use that the augmentations $R_t$
do not depend on the future.) So taking conditional expectation,
\begin{align*}
\E_{t-1}[\one(\tau \ge t) \cdot \ip{\lambda_t}{a_t(\psi_t(V_t) - x^*_t)}]
&~=~ \one(\tau \ge t) \cdot \ip{\lambda_t}{a_t\, \E_{t-1}[\psi_t(V_t) - x^*_t]} \\
&~=~ \one(\tau \ge t) \cdot \ip{\lambda_t}{a_t\, \E[\psi_t(V_t) - x^*_t]}
~=~0,
\end{align*} 
where the second equality follows from the fact that $\psi_t$ decides
based only on $V_t$, and hence is independent of the past.
Taking expectations and summing over $t$ gives $\E[M_{rhs}] = 0$.
Hence,
\begin{align}
  \E [\alg] \ge \E \min\bigg\{\sum_t C_t\; \psi_t(V_t) ~,~ 
  \frac{\gamma B}{8} \bigg\} ~- \gamma \cdot O(\log d).  \label{eq:prop-eq}
\end{align}

To prove \Cref{thm:prophetWithAug}, it suffices to show that the first
term in the minimization is $\Omega(\OPT_{base})$ with constant
probability. Notice that this term is always non-negative, and that
the second term in the minimization is $\Omega(\OPT_{base})$, since we
set $\gamma = \Theta(\frac{\OPT_{base}}{B}$).

\begin{lemma}
  \label{lem:cheb}
  With probability at least $0.6$ it holds that $\sum_t C_t \; \psi_t(V_t) \ge \Omega(\OPT_{base}).$
\end{lemma}
   
\begin{proof}
  First, we always have
  $\sum_t C_t \; \psi_t(V_t) \ge \sum_t V_t \; \psi_t(V_t)$, since
  $R_t \geq 0$. By the upper bound in
  Assumption \ref{assum:smoothProphet} on $V_t$, and the fact that
  $\psi_t(V_t) \in [0,1]$ and that for any random variable
  $X \in [0,\alpha]$ its variance satisfies $\Var(X) \le \alpha\E[X]$, we have
  \[
    \Var(V_t \,\psi_t(V_t)) ~\le~ \frac{\OPT_{base}}{20}\, \E [V_t\,\psi_t(V_t)].
  \]
  Using the fact that $\psi_t(V_t)$ are independent and \Cref{eq:relaxBndOPT} in \Cref{lem:prophBoundOPT}, we have
  \[
    \Var\Big(\sum_t V_t \psi_t(V_t)\Big) 
    	~\le~ \frac{\OPT_{base}}{20}\, \E\Big[\sum_t V_t\,\psi_t(V_t)\Big]
	~\le~ \frac{1}{10} \Big(\E\Big[\sum_t V_t\,\psi_t(V_t)\Big]\Big)^2.
  \]
  Applying Chebychev's inequality, we have
  \begin{align*}
    \Pr\Big(\sum_t V_t\,\psi_t(V_t) \le \tfrac{1}{2} \E\Big[\sum_t V_t\,\psi_t(V_t)\Big] \Big) ~\le~ 0.4.
  \end{align*}
  Finally, using  $\E[\sum_t V_t\,\psi_t(V_t)] \ge \frac{\OPT_{base}}{2}$ by \Cref{eq:relaxBndOPT} in \Cref{lem:prophBoundOPT} concludes the proof.  
\end{proof}

Combining \Cref{lem:cheb} with~(\ref{eq:prop-eq}) shows the expected
value of the algorithm is at least a constant fraction of the
$\OPT_{base}$, and hence proves~\Cref{thm:prophetWithAug}.


\iffull


	\section{Refined Approximation for Byzantine Packing Integer Programs}
        \label{sec:one-plus-eps}
	
	We consider again the Byzantine Packing Integer Program problem $\textrm{LP}(S)$ from \S\ref{sec:byzant-PIP}. In contrast to the constant approximation presented in that section, here we design a $(1-O(\e + \alpha))$-approximation when only an $\alpha < 1$ fraction of the items is red and a weak estimate of $\OPT(\Gitems)$ is available. 
	
\begin{theorem} \label{thm:combo2}
	Consider $\e \in (0,\frac{1}{10}]$ and $\delta \in [0,1]$ and assume $B \ge \Omega\big(\frac{\parts \log (d \parts/\delta)}{\e^2}\big)$ and $\parts \ge
  \Omega(\frac{\log\!\log n + \log 1/\e}{\e^2})$. Suppose the total value of the green items of value $> \frac{\OPT}{B}$ is at most $\e \OPT$, and that it is available an estimate $\widehat{\OPT} \in [\frac{1}{n} \OPT, n \OPT]$. Then the solution returned by \Cref{alg:AD-byz2} satisfies the following:
  \begin{OneLiners}
  \item[i.] (Feasibility) The solution always packs into the modified budget of $(B + \parts)\cdot\ones$.
  
  \item[ii.] (Value) The solution has value at least $(1-\alpha-O(\e)) \OPT$ with probability at least $1-\delta'$, where $\alpha$ is the fraction of red items in the instance. 
  \end{OneLiners}
\end{theorem}

	Again notice that we can scale down the solution found by the algorithm by $\frac{B}{B+K} = 1 - O(\e^2)$ to get a feasible solution with the same approximation guarantee.
	
	The algorithm is very similar 
        to the one used for the constant-factor approximation in
        \S\ref{sec:byzant-PIP}, but: 1) we now use a slightly modified
        version of Algorithm \ref{alg:AD-oneint} for each interval,
        and 2) we use a finer discretized set $\Gamma$ of possible
        $\gamma$'s. We now make these modifications precise and prove \Cref{thm:combo2}. 
	
	
	\subsection{Modified Version of \Cref{alg:AD-oneint}}

	As the ``within-a-single-interval'' algorithm we use the original Agrawal-Devanur algorithm for Packing LP's in the stochastic model~\cite{AD15} applied to items in an interval $I$ with the scaled budget $|I| B$. Using again the r.v.s $(C_t, A_t)$ to denote the value and size of the $t$-th item that appears in this interval, the main difference between this new algorithm and \Cref{alg:AD-oneint} is that it now uses the penalty functions 
	\begin{align*}
		\textstyle \tilde{f}_t(\lambda) := \ip{\lambda}{A_t x_t \blue{- \frac{B \cdot \ones}{n}}},
	\end{align*}
	(recall that $n$ is the total number of items in the whole
        instance) and, because some of the penalties can be negative, it
        considers $\lambda$ to be in the ``full-dimensional simplex''
        that includes the 0 vector, namely $\blacktriangle^d := \{ \lambda \in [0,1]^d : \sum_i \lambda_i \le 1\}$. The motivation for the additional term $- \frac{B \cdot \ones}{n}$ is that when there are few red items it helps canceling out the ``first order term'' of the green $\OPT$'s occupation given that $A_t x^*_t \le \frac{B \ones}{G} \approx \frac{B \ones}{n}$. For convenience, the algorithm is presented in \Cref{alg:AD-oneint2}.

\begin{algorithm}
  \caption{\blue{Modified} IntervalByzLP$(I,\gamma)$}
  \label{alg:AD-oneint2}
\begin{algorithmic}[1]
  \For{steps $t = 1, 2, \cdots$ in interval $I$}
    \State
    Use the low-regret OLO algorithm from \Cref{lem:OLOregret} on \blue{$\tilde{f}_1,\ldots,\tilde{f}_{t-1}$} to get
    $\dual_t \in \blue{\blacktriangle^d}$. 
    \State
    Compute $\algx_t \in [0,1]$ maximizing $x \mapsto C_t x - \gamma \ip{\dual_t}{A_t x}$.
    \State
    \textbf{break} if the scaled budget is violated, i.e., if a coordinate
    of $\sum_{s \leq t} A_s
    \algx_s$ exceeds $B\,|I|$.
  \EndFor
\end{algorithmic}
\end{algorithm}	

	The following is the  guarantee of this algorithm for interval $I$, with a proof similar to that of \Cref{lem:AD}. 

\begin{lemma} \label{lem:AD2}
	Under the assumptions from \Cref{thm:combo2}, for any interval $I \sse [0,1]$ with $|I| = \frac{1}{K}$ and $\gamma \in [0, \frac{\OPT}{B}]$, with probability at least $1 - \delta$ we have
  \begin{align*}
    \alg(I,\gamma) ~\ge~ \frac{1}{K} \cdot \min\bigg\{\gamma B\,,\, \OPT - \alpha \gamma B \bigg\} ~-~ O\bigg(\frac{\e \OPT}{K} + \frac{\e \gamma B}{K} + \frac{\e \gamma B n_I}{n} \bigg) \enspace. 
  \end{align*}

\end{lemma}

\begin{proof}
Let $G(I)$ be the number of green items in the interval $I$, and let $n_I$ be the total number of items in $I$. Again let $\tau$ be the number of items seen when the algorithm stops. 
  
	Let $x^*$ be the optimal solution for the whole instance (not just on the interval $I$) considering only the green items that have individual value at most $\frac{\OPT}{B}$. Recall that by the assumption $x^*$ has value at least $(1-\e) \OPT$. By the ``best response'' definition of $x_t$ we have
  \begin{align*}
    \sum_{t \le \tau} C_t \algx_t - \gamma \sum_{t \le \tau} \ip{\lambda_t}{A_t \algx_t - \tfrac{B \ones}{n}} ~\ge~ \underbrace{\sum_{t \le \tau} C_t x^*_t - \gamma \sum_{t \le \tau} \ip{\lambda_t}{A_t x^*_t - \tfrac{B \ones}{n}}}_{=: \widetilde{\cL}(\gamma)}.
  \end{align*}
  Moreover, since for any $x \in [0,1]$ the absolute value of $\ip{\lambda_t}{A_t x - \tfrac{B \ones}{n}}$ can be checked to be at most $\ip{\lambda_t}{A_t x - \tfrac{B \ones}{n}} +  \frac{2B}{n}$, by the regret guarantee of the  OLO algorithm have 
  \begin{align*}
  	\sum_{t \le \tau} \ip{\lambda_t}{A_t \algx_t - \tfrac{B \ones}{n}} &\ge (1-\e) \max_{\lambda \in \blacktriangle^d}\, \sum_{t \le \tau} \ip{\lambda}{A_t \algx_t - \tfrac{B \ones}{n}} - \e \tau \frac{2 B}{n} - \frac{\log d}{\e}\\
  	&\ge (1-\e) \max_{\lambda \in \blacktriangle^d}\, \sum_{t \le \tau} \ip{\lambda}{A_t \algx_t - \tfrac{B \ones}{n}} - O\bigg(\e \frac{B n_I}{n} + \e \frac{B}{K} \bigg),
  \end{align*}
  where the last inequality uses the fact $B \ge \Omega(\frac{K \log d}{\e^2})$. Combining these inequalities we get
  \begin{align}
    \alg(I,\gamma) ~\ge~ 
    \widetilde{\cL}(\gamma) \, +\, \gamma (1-\e) \max_{\lambda \in \blacktriangle^d}\, \sum_{t \le \tau} \ip{\lambda}{A_t \algx_t - \tfrac{B \ones}{n}} - O\bigg(\e \frac{\gamma B n_I}{n}\bigg) - O\bigg(\e \frac{\gamma B}{K} \bigg). \label{eq:packEps}
  \end{align}  

	The next step is to lower bound the ``stopped Lagrangian'' $\widetilde{\cL}(\gamma)$. Before that we need a simple concentration on the number of green items on interval $I$. 
  
  \begin{lemma} \label{lemma:concGI}
  With probability at least $1-\frac{\delta}{4}$ we have $G_I \in [(1-\e) \frac{G}{K}, (1+\e) \frac{G}{K}]$.
  \end{lemma}
  
  \begin{proof}
  The expected number of green items in $I$ is $\E G(I) = \frac{G}{K}$, and $\Var(G(I)) \le \E G(I)$ (since $G(I)$ is the sum of independent random variables in $[0,1]$). Then Bernstein's inequality (\Cref{lem:Bern}) we have
		\begin{align}
			\Pr(|G(I) - G/K| > \e G/K) \le 2 \exp\bigg(- \frac{\frac{\e^2 G^2}{2K^2}}{\Var(G(I)) + \frac{\e G}{3 K}} \bigg) \le 2 \exp\bigg(- \frac{3}{8} \frac{\e^2 G}{K}  \bigg) \le \frac{\delta}{4}, \label{eq:concGI}
		\end{align}
		where the last inequality is because $G \ge (1-\e) B$ (since by assumption the green items of value $>\frac{\OPT}{B}$ have total value at most $\e \OPT$ there must be at least $(1-\e) B$ other green items to obtain the remaining $(1-\e) \OPT$ value) and $B \ge \Omega(\frac{K \log 1/\delta}{\e^2})$
  \end{proof}
	
	Now we proceed with lower bounding
        $\widetilde{\cL}(\gamma)$. For that, let $\tilde{\tau}$ denote
        the number of green items in the interval $I$ before (and including) time step $\tau$.  

  \begin{lemma} \label{lemma:lag}
  	With probability at least $1-\frac{\delta}{2}$, 
  	\begin{align*}
  		\widetilde{\cL}(\gamma) ~\ge~ \frac{\tilde{\tau}}{G} \bigg[ \OPT - \alpha \gamma B \bigg] - O\bigg(\e \frac{\OPT + \gamma B}{K}\bigg).
  	\end{align*}
  \end{lemma}

	\begin{proof}
	  Recall that $\Gsize(I)$ denotes the number of green items that fall in interval $I$, and let $t_j$ be the time step of the $j^{th}$ green item to appear in  interval $I$ (notice that the $t_j$'s are random). Then 
	  \begin{align}
	  	 \widetilde{\cL}(\gamma) ~=~ \sum_{t \le \tau} C_t x^*_t - \gamma \sum_{t \le \tau} \ip{\lambda_t}{A_t x^*_t - \tfrac{B \ones}{n}} ~=~ \sum_{j \le \tilde{\tau}} C_{t_j} x^*_{t_j} - \gamma \sum_{j \le \tilde{\tau}} \ip{\lambda_{t_j}}{A_{j_t} x^*_{t_j}} - \gamma \sum_{t \le \tau} \ip{\lambda_t}{\tfrac{B \ones}{n}} . \label{eq:packHardMain}
	  \end{align}
	 We lower bound the first two terms of the right-hand side with high probability. 
	 
	 For that we condition on the number of green items on this interval $\Gsize(I) = N$ where $N \le \frac{2 G}{K}$. Also condition on the positions of the green items $(t_1,\ldots,t_N) =: t_{\le N}$. Notice that under this conditioning the stochastic items $(C_{t_j}, A_{t_j})$ are still sampled without
  replacement from the green items. Therefore, the conditional expected revenue from one (any) time step $t_j$
  \begin{align*}
  	\mu \,:=\, \E[C_{t_j} x^*_{t_j} \mid \Gsize(I) = N, t_{\le N}]
  \end{align*} 
  takes value between $(1-\e) \frac{\OPT}{G}$ and $\frac{\OPT}{G}$, and the conditional expected occupation 
  \begin{align*}
    \vec{\mu} := \E [A_{t_j} x^*_{t_j} \mid \Gsize(I) = N,\, t_{\le N} ]
  \end{align*}
  is coordinate-wise at most $\frac{B \ones}{G}$. 
  
  Moreover, the revenue is concentrated around the expectations: using
  the maximal Bernstein inequality \Cref{lemma:maximalBernstein}
  conditionally (with $X_j := C_{t_j} x^*_{t_j}$ and $M =
  \frac{\OPT}{B}$), we have
		\begin{align*}
			\Pr\Big(\max_{\ell \le N} \big|\tsum_{j \le \ell} (C_{t_j} x^*_{t_j} - \mu)\big| \ge \tfrac{\e c \OPT}{K} ~\Big|~ G(I) = N, t_{\le N} \Big) 
			&~\le~ 30 \exp\Big(- \tfrac{\e^2 B}{K} \cdot \min\big\{\tfrac{2 G}{K N}\,,\,1\big\} \Big) \notag ~\le~ \tfrac{\delta}{8}, 
		\end{align*}
	for a suitably large constant $c$. Here the last inequality
        uses $B \ge \Omega(\frac{K \log 1/\delta}{\e^2})$ and $N \le
        \frac{2G}{K}$. Notice that since $\tilde{\tau} \le N$ this
        also implies concentration for the sum $\sum_{j \le \tilde{\tau}} (C_{t_j} x^*_{t_j} - \mu)$ with random range $j \le \tilde{\tau}$. 
	
	In addition, the occupation is also concentrated: applying \Cref{lemma:breakCorr} conditionally we get 
		\begin{align*}
			\Pr\bigg(\sum_{j \le\tilde{\tau}} \ip{\dual_{t_j}}{ A_{t_j} x^*_{t_j}} > (1+4 \e) \sum_{j \le \tilde{\tau}} \ip{\lambda_t}{\vec{\mu}}  + \Omega\big(\tfrac{\log d/\delta}{\e}\big)  ~\bigg|~  G(I) = N, t_{\le N}\bigg) ~\le~ \frac{\delta}{8}.
		\end{align*}
	Notice we can indeed apply \Cref{lemma:breakCorr} because   $\dual_{t_j}$ is a function of the green items before the
  $j^{th}$ green item in the interval, namely 
  $(C_{t_1}, A_{t_1}), \ldots, (C_{t_{j-1}}, A_{t_{j-1}})$ (it also depends on the
  red items in $I$, but these are deterministic), and similarly $\tilde{\tau}$ is a stopping time with respect to the sequence $((C_{t_j}, A_{t_j}))_j$ and $\tilde{\tau} \le G(I) = N \le \frac{2 G}{K} \le \frac{G}{2}$. Taking a union bound over both concentration inequalities and then expectation over $t_{\le N}$ and over $G(I) \le \frac{2G}{K}$ and using the fact that $B \ge \Omega(\frac{K \log d/\delta}{\e^2})$ to upper bound the term $\Omega\big(\frac{\log d/\delta}{\e}\big)$ we get 
		\begin{align*}
			\Pr\bigg(\sum_{j \le \tilde{\tau}} C_{t_j} x^*_{t_j} - \gamma \sum_{j \le\tilde{\tau}} \ip{\dual_{t_j}}{ A_{t_j} x^*_{t_j}} < \tilde{\tau} \mu - (1+4 \e) \gamma \sum_{j \le \tilde{\tau}} \ip{\lambda_t}{\vec{\mu}}  + \Omega\big(\e \tfrac{\OPT + \gamma B}{K} \big)  ~\bigg|~  G(I) \le \frac{2G}{K}\bigg) \le \frac{\delta}{4}.
		\end{align*}

	From \Cref{lemma:concGI}, $G(I) \le \frac{2 G}{K}$ with probability at least $1 - \frac{\delta}{4}$, and so taking a union bound with the displayed inequality and employing it on \eqref{eq:packHardMain} we get  that with probability at least $1 - \frac{\delta}{2}$ we have $G(I) \le \frac{2 G}{K}$  and 
	\begin{align*}
		\widetilde{\cL}(\gamma) &\ge \tilde{\tau} \mu - (1+4 \e) \gamma \sum_{j \le \tilde{\tau}} \ip{\lambda_t}{\vec{\mu}}   - \gamma \sum_{t \le \tau} \ip{\lambda_t}{\tfrac{B \ones}{n}} + \Omega\big(\e \tfrac{\OPT + \gamma B}{K} \big)\\
		&\ge (1-\e) \tilde{\tau} \tfrac{\OPT}{G} - (1+4\e) \gamma \sum_{j \le \tilde{\tau}} \ip{\lambda_t}{\tfrac{B \ones}{G}} - \gamma \sum_{j \le \tilde{\tau}} \ip{\lambda_t}{\tfrac{B \ones}{n}} + \Omega\big(\e \tfrac{\OPT + \gamma B}{K} \big)\\
		&\ge (1-\e) \tilde{\tau} \tfrac{\OPT}{G} - \tilde{\tau} \gamma B \big[\tfrac{1+4\e}{G} - \tfrac{1}{n}\big] + \Omega\big(\e \tfrac{\OPT + \gamma B}{K} \big)\\
		&\ge (1-\e) \tilde{\tau} \tfrac{\OPT}{G} - \tilde{\tau} \gamma B \big[\tfrac{1+4\e}{G} - \tfrac{1-\alpha}{G}\big] + \Omega\big(\e \tfrac{\OPT + \gamma B}{K} \big) ,
	\end{align*}
	where the second follows inequality because $\ip{\lambda_t}{\ones}\le 1$ and the last inequality follows from $G \ge (1-\alpha) n$, since by assumption at most an $\alpha$ fraction of the $n$ items are red. Finally, since $G(I) \le \frac{2 G}{K}$ we have $\frac{\tilde{\tau}}{G} \le \frac{G(I)}{G} \le \frac{2}{K}$, and so all terms depending on $\e$ can be absorbed by the last term $O(\e \frac{\OPT + \gamma B}{K})$. The displayed inequality then gives the desired result. 
	\end{proof}

	Let $\cE$ be the event where \Cref{lemma:lag,lemma:concGI} hold, so by a union bound $\Pr(\cE) \ge
        1-\delta$. Fix a scenario in $\cE$, and break the analysis into two cases depending whether
        the algorithm exhausts a budget or not. 
    
	
	\paragraph{Case 1: $\tau < n_I$.} In this case the algorithm exhausted a budget $|I|B$ and hence
	\begin{align*}
		\max_{\lambda \in \blacktriangle^d} \, \sum_{t \le \tau} \ip{\lambda}{A_t \algx_t - \tfrac{B \ones}{n}} ~\ge~ |I| B - \frac{\tau B}{n} \,=\, \frac{B}{K} - \frac{\tau B}{n}.
	\end{align*}
	Since we are in a scenario in $\cE$, \eqref{eq:packEps} gives
  \begin{align*}
    \alg(I,\gamma) ~&\ge~ \frac{\tilde{\tau}}{G} \bigg[ \OPT -   \alpha \gamma B\bigg] + (1-\e) \frac{\gamma B}{K} - 2\e  \frac{\gamma B \cdot n_I}{n} - O\bigg(\e \frac{\OPT + \gamma B}{K} \bigg)\\
    &\ge~ \frac{\gamma B}{K} - O\bigg(\frac{\e \OPT}{K} + \frac{\gamma B}{K} + \frac{\e \gamma B \cdot n_I}{n} \bigg),
  \end{align*}
  where the second inequality follows from the fact $\gamma \le \frac{\OPT}{B}$, which makes $ \frac{\tilde{\tau}}{G} [ \OPT -   \alpha \gamma B] \ge 0$ and so we can just drop this term. This gives the first term in the minimum in \Cref{lem:AD2}.
  

	\paragraph{Case 2: $\tau = n_I$.} Since we can take $\lambda = 0$, we have 
	\begin{align*}
		\max_{\lambda \in \blacktriangle^d}\, \sum_{t \le \tau} \ip{\lambda}{A^t \algx_t - \tfrac{B \ones}{n}} \ge 0.
	\end{align*}	
	Moreover, in the current case (and scenario in $\cE$) we have $\tilde{\tau} = G_I \ge (1-\e) \frac{G}{K}$. Thus, using \Cref{lem:AD2} and inequality \eqref{eq:packEps} we get
  \begin{align*}
    \alg(I,\gamma) ~\ge~ \frac{(1-\e)}{K} \bigg[ \OPT - \alpha \gamma B\bigg] - \e \frac{\gamma B n_I}{n} + O\bigg(\e \frac{\OPT + \gamma B}{K} \bigg). 
  \end{align*}  	
  This gives the second term in the minimum in \Cref{lem:AD2}, and concludes the proof of the lemma.
\end{proof}


\subsection{Modified Version of \Cref{alg:AD-byz}}

	The only modification to the ``across intervals'' \Cref{alg:AD-byz} from \S\ref{sec:optUnknown} is to run it using the previous ``within-a-single-interval''  \Cref{alg:AD-oneint2} and to use a finer grid $\Gamma_\e$	of size $O(\frac{\log n}{\e})$ that contains $\frac{\OPT}{B}$ within a factor of $(1-\e)$, constructed based on the assumed estimate of $\OPT$ within $\poly(n)$ factors. This is stated more precisely in \Cref{alg:AD-byz2}.

\begin{algorithm}
  \caption{\blue{Modified} ByzLP}
  \label{alg:AD-byz2}
\begin{algorithmic}[1]
  \For{interval $i = 1, \ldots, K$}
    \State
    pick ${\gamma}_i \in \blue{\Gamma_\e}$ using the multiscale experts
    algorithm of \Cref{thm:multiscale} on 
    $\bar{\alg}_1,
    \bar{\alg}_2,\ldots, \bar{\alg}_{i-1}$.
    \State
    run \blue{\Cref{alg:AD-oneint2}} over interval $I_i$
    with parameter ${\gamma}_i$, thereby getting
    value $\alg_i({\gamma}_i)$.
  \EndFor
\end{algorithmic}
\end{algorithm}	
	
	We finally prove that the above algorithm has the approximation guarantee stated in \Cref{thm:combo2}.

\begin{proof}[Proof of \Cref{thm:combo2}]
 Again the feasibility follows directly by the feasibility guarantee of \Cref{alg:AD-oneint2}.
 
 Let $\gamma^*$ be a value in $\Gamma$ in the interval $\big[(1-\e) \frac{\OPT}{B}, \frac{\OPT}{B}\big]$. By the truncation we get that $\bar{\alg}_i(\gamma^*) \le \frac{\OPT}{K}$. Thus, the multiscale regret guarantee and the fact $K \ge \Omega(\frac{\log |\Gamma_\e|}{\e^2})$ gives that in every scenario the total value of the algorithm is at least 
  \begin{align}
    \sum_i \bar{\alg}_i({\gamma}_i) ~&\ge~ (1-\e) \sum_i \bar{\alg}_i(\gamma^*) - \OPT \cdot O\bigg(\frac{\log |\Gamma_\e|}{\e \parts}\bigg) \ge (1-\e) \sum_i \bar{\alg}_i(\gamma^*) - O(\e \OPT). \label{eq:packingWrapEps1}
  \end{align}
  
  With probability at least $1-\delta'$, the guarantee from \Cref{lem:AD2} (with $\delta$ set as $\frac{\delta'}{\parts}$) holds for all $\parts$ intervals. Whenever this holds we have
  \begin{align*}
    \bar{\alg}_i(\gamma^*) ~&\ge~ |I_i| \min\bigg\{(1-\e) \OPT \,,\, (1-\alpha) \OPT \bigg\} ~-~ O\bigg(\frac{\e \OPT}{K} + \frac{\e \OPT n_{I_i}}{n} \bigg)\\
    & \ge~ |I_i| (1 - \alpha - O(\e)) \OPT - O\bigg(\frac{\e \OPT n_{I_i}}{n}\bigg)
  \end{align*}
  for all $i \in [\parts]$. Adding over all $i$'s then gives $\sum_i \bar{\alg}_i(\gamma^*) \ge (1-\alpha - O(\e)) \OPT$. Applying this on \eqref{eq:packingWrapEps1} then concludes the proof of the theorem. 
\end{proof}


\fi


\section{Byzantine Secretary for Value Maximization}
\label{sec:valueMaxi}


In this section we consider two classical secretary problems in the
Byzantine framework: that of (i) picking a single item, and (ii)
picking an independent set in a matroid, to maximize the expected
value of picked items. The main results of this section are:

\ValueMax*

The first claim is proved in \S\ref{sec:value-max}, and improves upon
the previous factor of $\Omega\big(\nicefrac{1}{(\log^* n)^2}\big)$ of \cite{BGSZ-ITCS20}. The second
claim is then proved in \S\ref{sec:matroids} 
and improves upon the previous factor of $\Omega\big(\nicefrac{1}{\log n}\big)$ when the rank
$r$ is sufficiently smaller than the total number of items $n$, also from~\cite{BGSZ-ITCS20}.

Note that these packing problems are related to the results we saw in
\S\ref{sec:byzant-PIP}
\iffull
and \S\ref{sec:one-plus-eps}
\fi: both the current
problems can also be modeled as PIPs, but we cannot assume that the capacities
(i.e., the right-hand sides of the PIP) are large, so we cannot apply
results we proved earlier. The single-item \emph{probability
  maximization} result from \S\ref{sec:single-item} also gives an
$\Omega(1/\log n)$-approximation for the value-maximization setting,
which is much weaker than the result here.

\subsection{Single Item Value Maximization}
\label{sec:value-max}


For the value-maximization problem, we give three procedures, each of
which picks at most $K=O(\log^* n)$ items. By picking all items chosen
by any of the three procedures, we get an algorithm that picks
$O(\log^* n)$ items and gets expected value $\Omega(\target)$. Then
picking one of these items uniformly at random, we get a single item
with expected value $\Omega(\target/\log^* n)$. 

The main idea of this algorithm is break the time horizon into $K$
intervals, and to iteratively get a better estimate of the optimum
value each time we fail to get a large expected value.  As in
\S\ref{sec:prob-max}, we divide time into $K$ intervals, and let
$\widehat{\mu}_i$ be the largest value of any item in interval
$I_i$. For the sake of intuition, assume that
$n\target\ge \widehat{\mu}_0 \ge \widehat{\mu}_1 \ge \dots \ge
\widehat{\mu}_K \ge \target$. Suppose we know that $\widehat{\mu}_i$
is an $\alpha$-approximation to $\target$. Pick $k$ randomly from
$[\log \alpha]$ and set the threshold $2^{-k}\widehat{\mu}_i$ for
interval $I_{i+1}$. Since we know that $\widehat{\mu}_{i+1}$ is
contained in the interval
$[2^{-k}\widehat{\mu}_i, 2^{-(k-1)}\widehat{\mu}_i]$ for some
$k\in [\log n]$, choosing this value of $k$ yields value
$\approx \widehat{\mu}_{i+1}$ with probability at least
$\frac{1}{\log \alpha}$. If
$\widehat{\mu}_{i+1} \ge \target \log \alpha $, then we get expected
value $\target$. Otherwise, we know that $\widehat{\mu}_{i+1}$ is a
$\log \alpha$-approximation to $\target$. Since $\widehat{\mu}_0$ is
an $n$-approximation to $\target$, we find inductively that each
$\widehat{\mu}_i$ is a $\log^{(i)} n$-approximation to $\target$. This
implies that we need consider only $K=O(\log^* n)$ many intervals.
This gives the main idea: the actual procedure needs to consider
situations where the interval maxima are not monotonically decreasing,
so we need more care.

Formally, the structural procedures are \Struct---defined in
\Cref{procedure:struct} of \S\ref{sec:single-item}---and a second procedure which simply picks
one random item.

\begin{algorithm}
  \caption{Procedure \textsc{Sample}}
  \label{alg:sample}
\begin{algorithmic}[1]
  \State
  Choose one item uniformly at random
\end{algorithmic}
\end{algorithm}

To explain the third and main search procedure, let us recall useful
notation from \S\ref{sec:single-item}: let $\rho_i$ be the maximum
value of any red item in interval $I_i$; define $\rho_i \gets -\infty$
if no such items exist. An interval $I_i$ (including $I_0$) is
\emph{high} if $\rho_i \ge \target$, and \emph{low} otherwise. An
interval $I_i$ (for $i\ge 1$) is \emph{\good} if the intervals $I_i$
and $I_{i-1}$ are both high, and moreover the item $\gmax$ does not
arrive in interval $I_{i-1}$.

Now in interval $I_i$, the new search procedure determines the
location of the previous high intervals, assuming that $I_{i-1}$ is
high.  
The threshold for
interval $I_i$ is chosen from among $\log^{(\ell+1)} n$
exponentially-spaced values centered around the $\widehat{\mu}$ value
of the previous {\good} interval, where $\ell$ is roughly the number
of previous {\good} intervals.

\begin{algorithm}
  \caption{Procedure \SearchII}
  \label{}
\begin{algorithmic}[1]
  \For{interval $i = 1, 2, \ldots, K$}
  \State
    $\Good_i \gets \{j\in [i-1] : \widehat{\mu}_j \ge \widehat{\mu}_{i-1} \text{ and } \widehat{\mu}_{j-1} \ge \widehat{\mu}_{i-1}\}$
  \State
    $\widehat{\ell}_i \gets |\Good_i|+1$
    \If{$\widehat{\ell}_i > 1$}
      \State{$\hat{\jmath}_i \gets \max(\Good_i)$}
      \Else
      \State $\hat{\jmath}_i \gets i-1$
    \EndIf
    \State
    $\widehat{k}_i\gets$ uniformly random integer in $[-\log^{(\widehat{\ell}_i)}n, \log^{(\widehat{\ell}_i)}n]$
    \State
    In interval $I_i$, pick the first item (if any) with value at least $\tau_i := 2^{\widehat{k}_i}\cdot \widehat{\mu}_{\hat{\jmath}_i}$.
  \EndFor
\end{algorithmic}
\end{algorithm}

Intuitively, $\Good_i$ is the index set of {\good} intervals among $I_1,\dots, I_{i-1}$, the number $\widehat{\ell}_i$ is one more than the number of previous {\good} intervals, and $\hat{\jmath}_i$ is the index of the most recent {\good} interval. Indeed, these statements are all true when $I_{i-1}$ is high (\Cref{lem:great-sequence}).

\paragraph{The  Analysis.}
It is clear that each procedure picks at most $K = O(\log^*n)$ items.

\begin{lemma}\label{lem:struct-search-good}
\Struct has expected value $\Omega(\target)$ unless the following properties both hold:
\begin{OneLiners}
\item[(i)] If $I_i$ and $I_j$ are high intervals with $i<j$, then $\rho_i > \rho_j$.
\item[(ii)] There are at most $\frac{K}{4}$ low intervals.
\end{OneLiners}
\Sample has expected value $\target$ unless the following property holds:
\begin{OneLiners}
\item[(iii)] The maximum value of any item is less than $n\target$.
\end{OneLiners}
\end{lemma}
\begin{proof}
By \Cref{lem:alg1-good}, if either of $(i)$ or $(ii)$ fails then \Struct picks an item of value at least $\target$ with probability $\Omega(1)$, and hence has expected value $\Omega(\target)$. 
The second claim is immediate.
\end{proof}

Let $j_1,\dots, j_m$ denote the indices of the {\good} intervals. Let $\GoodInt_\ell := I_{j_\ell}$ denote the $\ell$-th {\good} interval and let $\widehat{\mu}(\Good_\ell) := \widehat{\mu}_{j_\ell}$ denote its most valuable item. Let $\GOOD_i = \{j_\ell \mid j_\ell < i\}$ denote the set of indices of {\good} intervals before interval $I_i$. 

\begin{lemma}\label{lem:great-sequence}
Suppose that item $\gmax$ arrives in $(\frac{3}{4},1]$ and property (i) of \Cref{lem:struct-search-good} holds. If interval $I_{i-1}$ is high, then $\Good_i=\GOOD_i$. In particular, if $I_i$ is {\good}, then $I_{i} = \GoodInt_{\widehat{\ell}_i}$, and $I_{\hat{\jmath}_i} = \GoodInt_{\widehat{\ell}_i-1}$. \end{lemma}
\begin{proof}
Since $\gmax$ arrives in $(\frac{3}{4},1]$, for each $j$ it holds that $\rho_j = \widehat{\mu}_j$. Since $I_{i-1}$ is high, property (i) implies that $\widehat{\mu}_{i-1}$ is minimal among the $\widehat{\mu}$ values of high intervals seen so far. Thus, for $j\le i-1$, interval $I_j$ is {\good} if and only if $\widehat{\mu}_j \ge \widehat{\mu}_{i-1}$. It follows from definition that $\Good_i=\GOOD_i$.
\end{proof}

It is convenient to let $\GoodInt_0$ denote the (high) interval immediately preceding $\GoodInt_1$, so that for all $\ell=1,\dots, |\GOOD_i|$, in interval $I_{j_\ell}=\GoodInt_{\ell}$, \SearchII's choice of $\hat{\jmath}_{j_\ell}$ satisfies $I_{\hat{\jmath}_{j_\ell}} = \GoodInt_{\ell-1}$. 

\begin{lemma}\label{lem:algo3-good}
Let $K= 2\log^*n+2$. Suppose that (i)-(iii) of \Cref{lem:struct-search-good} all hold. Then \SearchII has expected value $\Omega(\target)$.
\end{lemma}
\begin{proof}
By property $(ii)$, the number of {\good} intervals is at least 
\[K - \tfrac{K}{4} - \tfrac{K}{4} - 1  ~=~ \tfrac{K}{2} -1 ~\ge~ \log^*n.\]
Let $\ell$ be the minimal number $l\ge 1$ satisfying $\widehat\mu(\GoodInt_l)
\ge \target\log^{(l)}n $. (This is well defined---the inequality is always satisfied by $l=\log^*n$.)
We claim that
\[\target ~\le~ \widehat\mu(\GoodInt_{\ell-1}) ~<~ \target \log^{(\ell-1)}n.\] 
The first inequality follows from the fact that $\GoodInt_\ell$ is {\good} (in
particular, it is high). When $\ell=1$, the second inequality holds by property $(iii)$; when $\ell>1$, it holds by the minimality of $\ell$.
Therefore, there is some $k\in [-\log^{(\ell)}n, \log^{(\ell)}n]$ such that 
\[2^k \widehat{\mu}(\GoodInt_{\ell}) ~\le~ \target \log^{(\ell)}n ~<~ 2^{k+1} \widehat{\mu}(\GoodInt_{\ell}).\]
For ease of notation, let $t = j_\ell$ so that $I_t = \GoodInt_{\ell}$. Condition on $\gmax$ arriving in $(\frac{3}{4},1]$---which happens with probability $\frac{1}{4}$---and consider the choice of \SearchII in interval $\GoodInt_\ell$. Since property $(i)$ also holds, \Cref{lem:great-sequence} applies. In particular, $\widehat{\ell}_t = \ell$ and $\widehat\mu_{\hat{\jmath}_t}=\widehat\mu_{j_{\ell-1}} = \widehat\mu(\GoodInt_{\ell-1})$. Therefore $\widehat{k}_t$ is drawn from $[-\log^{(\ell)}n, \log^{(\ell)}n]$, and the threshold $\tau_t = 2^{\widehat{k}_t}\cdot \widehat{\mu}(\GoodInt_{\ell-1})$. With probability $\Omega\Big(\frac{1}{\log^{(\ell)}n}\Big)$, the algorithm chooses $\widehat{k}_t=k$. In this case, since $\widehat{\mu}_t = \widehat{\mu}(\GoodInt_{\ell}) \ge \target\log^{(\ell)}n\ge \tau_t$, \SearchII will pick an item in $I_t$. Thus \SearchII gets expected value at least 
\[\tfrac{1}{4}\cdot \Omega(1)\cdot \tfrac{1}{\log^{(\ell)}n} \cdot  \target\log^{(\ell)}n ~=~ \Omega(\target). \qedhere \]  
\end{proof}

\begin{proof}[Proof of~\Cref{thm:mainValue}(i)]
  We now run one of the three algorithms \Struct, \Sample, and
  \SearchII uniformly at
  random. By~\Cref{lem:struct-search-good,lem:algo3-good} the picked
  algorithm gives expected value at least $\target$ with probability
  $\nf13$, while picking at most $K = O(\log^* n)$ items. Actually
  selecting a random one of these $K$ items proves the claim.
\end{proof}


\subsection{Byzantine Secretary for Matroids}
\label{sec:matroids}

We consider the setting where the items are the ground set of an
arbitrary matroid and the algorithm can choose any independent set of
items. The algorithm is given the rank $r$ of the matroid and the
number $n$ of items, but only learns the matroid structure as the
items arrive. The main result of this section is \Cref{thm:mainValue}(ii).
The factor of $O((\log^*n)^2\cdot \log r)$ improves upon the previous
factor of $O(\log n)$ if the rank of the matroid $r$ is much smaller
than the number of items $n$. For the rest of the section, we
assume the benchmark $\target:= \OPT(\Gitems\setminus\gmax)$ satisfies $\target \ge 2\val(\gSecMax)$. This is
without loss of generality: indeed, we can run the algorithm from
\S\ref{sec:value-max} with probability $\frac12$ (which suffices when
$\target \le 2\val(\gSecMax)$) and run the algorithm below with the remaining
probability.

\paragraph{The Algorithm.} The idea of our algorithm is to use the approach from
\S\ref{sec:value-max} that, at each step, it either gets items of high
value, or gets a better estimate of the highest-value items.  Define
intervals $I_0$ and $I_i$, and values $\rho_i$ and $\widehat{\mu}_i$
for each interval as in \S\ref{sec:single-item}. Again the algorithm will pick one of the following 3 procedures uniformly at random and use it for the whole instance.

\paragraph{Procedure 1.} Run the single-item value-maximization algorithm from \S\ref{sec:value-max}. 

\paragraph{Procedure 2.} The second procedure is adapted from the
$O(\log n)$-competitive algorithm of~\cite{BGSZ-ITCS20}. The intuition
is as follows: most of the value in $\OPT$ comes from items with value
in $[\val(\gSecMax)/2r, \val(\gSecMax)]$. If this value interval is partitioned into
exponentially-separated value-levels, an average level contains a
$\frac{1}{\log r}$ fraction of the value of $\OPT$. The procedure estimates $\val(\gSecMax)$ by the max-value item in a random interval, uses this estimate to choose a random-value interval, and sets a threshold to get the value of this interval. More precisely, the procedure is the following:

\begin{algorithm}
  \caption{Procedure \IBGSZ}
  \label{algo:IBGSZ}
\begin{algorithmic}[1]
  \State
  $i \gets$ uniformly random integer in $[K]$
  \State
  $k \gets$ uniformly random integer in $[\log (2r^2)]$
  \State
  $\widehat{\mu}_i = $ max value of an item arriving in $I_i$
  \State
  After $I_i$, greedily chose items with value at least $2^{-k}\widehat{\mu}_i$.
 \end{algorithmic} 
\end{algorithm}

\paragraph{Procedure 3.} The third procedure attempts to capture the value of large red items. It is simpler than \SearchII from \S\ref{sec:value-max}, but similar in spirit:
\begin{algorithm}
  \caption{Procedure \textsc{SearchIII}}
  \label{algo:searchIII}
\begin{algorithmic}[1]
\State
  $i \gets$ uniformly random integer in $[K]$
\State
  $k \gets$ uniformly random integer in $[-\log^{(i)}(n), \log^{(i)}(n)]$
\State
  $\widehat{\mu}_{i-1} = $ max value of an item arriving in $I_{i-1}$
\State
  Pick the first item arriving in (or after) $I_i$ with value at least $2^{k} \widehat{\mu}_{i-1}$. If no such item arrives, pick nothing.
\end{algorithmic}
\end{algorithm}

\paragraph{The  Analysis.}

Similarly to \S\ref{sec:single-item}, we analyze the algorithm by breaking into different cases based on the following structural properties of the instance. 

\begin{OneLiners}
\item[(i)] No item has value more than $n\cdot r\cdot \val(\gSecMax)$.
\item[(ii)] $\target \ge 4\val(\gSecMax)$.
\item[(iii)] For every $i$, it holds that $\rho_i \ge r\cdot \val(\gSecMax)$. 
\end{OneLiners}

\noindent Recall that $\rho_i$ denotes the maximum value of a red item arriving in interval $I_i$.
\newcommand{\singleitem}{the algorithm of \S\ref{sec:value-max}\xspace}
\begin{lemma}\label{lem:single-item-good}
If either of properties (i) or (ii) fail, then \singleitem gets expected value $\Omega(\frac{\target}{\log^*n})$.
\end{lemma}
\begin{proof}
Suppose property (i) fails. With probability $\Omega(\frac{1}{\log^*n})$, \singleitem picks a uniformly random item. This item has expected value at least $r\cdot \val(\gSecMax)\ge \target$, so \singleitem has expected value $\Omega(\frac{\target}{\log^*n})$.

Now suppose property (ii) fails. Then as proved in \S\ref{sec:value-max}, \singleitem has expected value $\Omega(\frac{\val(\gSecMax)}{\log^* n}) \ge \Omega(\frac{\target}{\log^* n})$.
\end{proof}

\begin{lemma}\label{lem:int-bgsz-good}
If properties (i) and (ii) hold and property (iii) fails, then \IBGSZ\xspace has expected value at least $\Omega(\frac{\target}{K^2\log r})$. 
\end{lemma}
\begin{proof}
Fix $i$ such that $\rho_i  \le r\cdot \val(\gSecMax)$. Condition on the event that (a) \IBGSZ\xspace chooses this $i$ in line 1 and (b) $\gSecMax$ arrives in $I_i$ but $\gmax$ does not, which happens with probability $\frac{1}{K}\cdot\frac{1}{2K}(1-\frac{1}{2K}) \ge \frac{1}{3K^2}$. Now $\widehat{\mu}_i = \max\{\rho_i, \val(\gSecMax2)\}$, which implies that 
\begin{equation}\label{eq:mu-good-estimate}
\val(\gSecMax) ~\le~ \widehat{\mu}_i ~\le~ r\cdot \val(\gSecMax).
\end{equation}
Let $OPT$ denote the items in the optimal solution of green items excluding $\gmax$.

Let $L^* = \big\{x\in OPT\mid \val(x) > \frac{\val(\gSecMax)}{2r}\big\}$. Since $|OPT| \le r$, the items in $OPT\setminus L^*$ have total value at most $\frac{\val(\gSecMax)}{2}\le \frac{\target}{2}$. Thus $\val(L^*) \ge \frac{\target}{2}$. Given the conditioning, all items in $L^*\setminus \{\gSecMax\}$ arrive at independent uniformly random times. In particular, since $I_i\sse [\frac{1}{4}, \frac{3}{4}]$, the expected value of items in $L^*$ that arrive after $I_i$ is at least $\frac{1}{4}(\frac{\target}{2} - \val(\gSecMax)) \ge \frac{\target}{16}$ (using property (ii)). The value of each of these items is in $[\frac{\val(\gSecMax)}{2r}, \val(\gSecMax)]$. By \eqref{eq:mu-good-estimate}, each of these items has value in some value-level $[2^{-k} \widehat{\mu}_i, 2^{-k+1} \widehat{\mu}_i)$ with $k\in [\log (2r^2)]$.
An average level therefore contains at least $\Omega(\frac{\target}{\log r})$ value from $OPT$. The greedy algorithm on one such value-level gets at least half of $OPT$'s value in that level, so \IBGSZ\xspace gets expected value $\Omega(\frac{\target}{\log r})$ as well.
Since this holds when we condition on an event of probability $\Omega(\frac{1}{K^2})$, the result follows.
\end{proof}

\begin{lemma}\label{lem:big-red-good}
Let $K= \log^*(n)$. If properties (ii) and (iii) hold, then procedure \textsc{SearchIII} gets expected value $\Omega(\frac{\target}{\log^* n})$
\end{lemma}
\begin{proof}
Fix the minimal index $i\ge 0$ for which $\mu_i \ge r\cdot
\log^{(i)}n \cdot \val(\gSecMax)$. Since property (iii) holds, the index
$i=\log^*(n)$ satisfies the desired constraint, hence $i$ is well-defined. Furthermore, since property $(i)$ holds, $i\ge 1$. Then we
have the following inequalities: 
\[
r\cdot \val(\gSecMax)~\le~ \mu_{i-1} ~\le~
  \log^{(i-1)}(n)\cdot r \cdot \val(\gSecMax),
  \] 
  the first by property (ii)
and the second by minimality of $i$. It follows that there is a $k\in
[\pm \log^{(i)}(n)]$ such that the value level defined by
$2^k\mu_{i-1}$ contains an item of value at least $r\cdot
\log^{(i)}n \cdot \val(\gSecMax)$. With probability $\frac{1}{K\cdot
  2\log^{(i)}(n)}$, Procedure~3 picks interval $i$ and level $k$. In
this case, it gets value at least 
\[
r\cdot \log^{(i)}n \cdot \val(\gSecMax)
  ~\ge~ \log^{(i)}n \cdot \target.
  \] 
  Thus, \textsc{SearchIII} gets expected value $\Omega(\frac{\target}{K}) = \Omega(\frac{\target}{\log^* n})$.
\end{proof}

We now complete the proof of \Cref{thm:mainValue}(ii).

\begin{proof}[Proof of \Cref{thm:mainValue}(ii)] 
Consider the algorithm that uses $K = \log^* n$ and runs one of the above three procedures, each with probability $\frac{1}{3}$. By \Cref{lem:single-item-good,lem:int-bgsz-good,lem:big-red-good}, at least one of these procedures has expected value $\Omega\big(\frac{\target}{(\log^*n)^2\cdot \log r}\big)$. With probability $\frac{1}{3}$, the algorithm picks the right procedure. Hence, the same value guarantee applies up to a constant.
\end{proof}


\bigskip
\bigskip
\appendix
\noindent {\LARGE \bfseries Appendix}

%

\section{Useful Prior Results}

Let $\triangle^{d-1} := \big\{ \mathbf{p} \in [0,1]^d \mid \| \mathbf{p}\|_1 = 1 \big\}$ be the
\emph{probability simplex} in $\R^d$, and $\blacktriangle^d := \big\{ \mathbf{p} \in [0,1]^d \mid \| \mathbf{p}\|_1 \leq 1 \big\}$ be the
\emph{full-dimensional} probability simplex. 
We recall the full-information
online linear optimization (OLO) low regret bound; see, e.g., \cite{AHK-Survey}. 

\begin{lemma}[OLO] \label{lem:OLOregret}
  Fix $\varepsilon \in (0,\nicefrac{1}{2}]$. 
  The  experts algorithm
  considers a setting with $d$ experts. At each time the algorithm
  plays a probability distribution $p^t \in \triangle^{d-1}$ and receives
  a linear reward function $f_t:\triangle^{d-1} \rightarrow [-1,1]$. 
  For any time $\tau$, let $p^* = \argmax_{p \in \triangle^{d-1}} \sum_{t \le \tau} f_t(p)$ 
  be fixed action that gives the best reward over the entire input sequence. 
  Until any time $\tau$, the following holds:
  \begin{align}
  	\sum_{t \le \tau} f_t(p^*) - \sum_{t\le \tau} f_t({p^t}) \leq \varepsilon \sum_{t \le \tau} |f_t(p^*)| + \frac{
    \log d}{\varepsilon} \enspace .
  \end{align}
\end{lemma}


We will also need the following generalization from Bubeck et
al.~\cite[Theorem~1]{BDHN}  for the
full-information multiscale online learning problem.

\begin{lemma}[Multi-scale regret of Bubeck et
al.~\cite{BDHN}]
  \label{thm:multiscale}
  Fix $\varepsilon \in (0,1]$. The multi-scale experts algorithm
  considers a setting with $M$ experts. At each time the algorithm
  plays a probability distribution $p^t \in \triangle_M$ and receives
  a reward vector $r^t$ with each $r^t_i \in [0,c_i]$; moreover, $c_i$
  is known in advance. Let $R_i := \max_i \sum_t r^t_i$ be the reward
  of action $i$ over the entire input sequence. The following holds
  for each $i \in [M]$:
  \begin{align}
    R_i - \sum_t \ip{r^t}{p^t} \leq \varepsilon R_i + O\Big(\frac{c_i \,
    \log M}{\varepsilon} \Big) \enspace .
  \end{align}
\end{lemma}
If we pick a random action in $[M]$ at each timestep $t$ independently
from the distribution $p^t$, then the above theorem gives a guarantee
for the expected regret against oblivious adversaries as well.

The following concentration inequality is classical: see, e.g.,
\cite[Theorem~2.8.4]{Ver18}:
\begin{lemma}[Bernstein's Inequality]
  \label{lem:Bern}
  For $X_1, X_2, \ldots, X_n$ independent mean-zero random variables
  such that $|X_i| \leq M$ for all $i$, and any $t \geq 0$,
  \[ \Pr\Big[ \Big| \sum_i X_i \Big| > t \Big] \leq 2 \exp \Big(
    \frac{ t^2/2}{ \sigma^2 + Mt/3 }\Big), \]
  where $\sigma^2 = \sum_i \E[ X_i^2 ]$ is the variance of the sum.
\end{lemma}

Our algorithm for Packing Integer Programs will use the following
result about Byzantine Knapsacks (and hence about the multiple-item
Byzantine secretary) from \cite[Theorem~2 and Lemma~7]{BGSZ-ITCS20} in
case most of the value is concentrated on a small number of items.
\begin{lemma}
  \label{lem:itcs-result}
  There is an algorithm for Byzantine secretary on knapsacks with size
  at least $B \geq poly(\e^{-1} \log n)$ (and items of at most unit
  size) that is $(1-\e)$-competitive with the benchmark
  $\OPT(\Gitems \setminus \gmax)$. Moreover, there is an algorithm
  that given an estimate that is at least $\OPT(\Gitems)$ and at most
  $\poly(n)$ times as much, and knapsack size
  $B \geq poly(\e^{-1} \log n)$, prodices a solution that is
  $(1-\e)$-competitive with the benchmark $\OPT(\Gitems)$.
\end{lemma}


%


\section{Handling Correlations due to Sampling Without Replacement} \label{app:breakCorr}

In contrast to i.i.d. arrivals, to handle correlations due to sampling without replacement 
for secretary problems, we will use the following general lemma. 

	\begin{lemma} \label{lemma:breakCorr} Consider a set of vectors
	  $\{y^1, \ldots, y^m\} \in [0,1]^d$ and let $Y^1, \ldots, Y^k$ be sampled without replacement from this set. Let $Z^1,\ldots,Z^k$ be random vectors in $\blacktriangle^d$ such that $Z^j$ is a (possibly random) function of $Y^1,\ldots,Y^{j-1}$ for all $j$. Let $\tau$ be a stopping time for the sequence $((Y^t,Z^t))_t$ such that $\tau \le \frac{m}{2}$. Then for any $\e \in (0,\frac{1}{10}]$ and $\delta \in (0,1]$, with probability at least
	  $1- \delta$ we have
	  \[ 
	  \sum_{j \le \tau} \ip{Z^j}{Y^j} ~\le~ (1 + 4\e)\,
	  \sum_{j \le \tau} \ip{\E_{j-1} Z^j}{\E Y^j} \,+\, \frac{O(\log \nicefrac{d}\delta)}{\e} \enspace ,
	\]
	where $\E_{j-1} Z^j = \E[Z^j \mid (Y_1,Z_1),\ldots,(Y_{j-1},Z_{j-1})]$. 
	\end{lemma}
Special cases of the above lemma have appeared before in the literature, e.g.  of \cite[Lemma~5]{GM-MOR16}.
	We will need the following convenient concentration inequality for ``martingales with drift''.
	
	\begin{lemma}[Lemma 2.2 of \cite{nikhilRounding}] \label{lemma:nikhil}
		Let $X_1,X_2,\ldots,X_k$ be a sequence of (possibly dependent) random variables with values in $(-\infty,1]$ and such that there is $\alpha \in (0,1)$ such that $$\E[X_j \mid X_1,\ldots,X_{j-1}] \le - \alpha \E[X^2_j \mid X_1,\ldots,X_{j-1}]$$ for all $j$. Then for all $\lambda \ge 0$
		\begin{align*}
			\Pr(X_1 + \ldots + X_k > t) \le e^{-\alpha \lambda}.
		\end{align*}
	\end{lemma}

	We also need a maximal Bernstein's inequality for sampling without replacement.It follows by applying Lemma 1 of \cite{GM-MOR16} to the scaled random variables $\frac{X_i}{M} \in [0,1]$ and using the fact that $\Var(X) \le \E X$ for every random variable in $[0,1]$ (the last inequality follows from the inequality $\frac{a}{b+c} \ge \min\{\frac{a}{2b}, \frac{a}{2c}\}$, valid for all non-negative reals $a,b,c$).

	\begin{lemma}[Lemma 1 of \cite{GM-MOR16}]  \label{lemma:maximalBernstein}
		Consider a set of real values $x_1, \ldots, x_m$ in $[0,M]$, and let $X_1, \ldots, X_k$ be sampled without replacement from this collection. Assume $k \le m/2$. Let $S_i = X_1 + \ldots X_i$. Also let $\mu = \frac{1}{m} \sum_i x_i$ and $\sigma^2 = \frac{1}{m} \sum_i (x_i - \mu)^2$. Then for every $\alpha > 0$
		\begin{align*}
			\Pr\left(\max_{i \le k} |S_i - i \mu| \ge \alpha\right) \le 30 \exp\left(-\frac{(\alpha/24)^2}{M(2 k \mu + (\alpha/24))} \right) \le 30 \exp\left(-\min\bigg\{\frac{(\alpha/24)^2}{4 k \mu M}~,~ \frac{\alpha}{48 M} \bigg\} \right)
		\end{align*}
	\end{lemma}

	Let $\cF_j$ be the $\sigma$-algebra generated by $Y^1,\ldots,Y^j$ and $Z^1,\ldots,Z^j$, i.e., the history up to time $j$. We use $\E_{j-1}[\cdot] := \E[\ \cdot  \mid \cF_{j-1}]$ to denote expectation conditioned on the history up to time $j-1$.

	\begin{lemma} \label{lemma:balancedEveryTime}
		Consider $i \in [d]$. Then with probability at least $1 - \frac{\delta}{d}$ we have $\E_{j-1} Y^j_i \le (1+2\e) \E Y^j_i + \frac{O(\log d/\delta)}{m \e}$ for all $j \le \frac{m}{2}$ simultaneously. 
	\end{lemma}
	
	\begin{proof}
	Let $\mu = \frac{1}{m} \sum_{j \le m} y^j_i$, which is the expected value of $Y^j_i$. Moreover, the conditional expectation $\E_{j-1} Y^j_i$ is the average of the $y^t_i$'s that have not appeared up until time $j-1$, namely 
	\begin{align}
	\E_{j-1} Y^j_i = \frac{\sum_t y^t_i - \sum_{t \le j-1} Y^t_i}{m - (j-1)} = \frac{m \mu - \sum_{t \le j-1} Y^t_i}{m - (j-1)}. \label{eq:expOhat}
	\end{align}
	We then bound the last term uniformly for all $j \le \frac{m}{2}$ using the maximal Bernstein's inequality Lemma~\ref{lemma:maximalBernstein}.

	For that let $\sigma^2 := \frac{1}{m} \sum_t (y^t_i - \mu)^2$ and notice that 
	\begin{align*}
		\sigma^2 \,=\, \frac{1}{m} \sum_t (y^t_i)^2 - \mu^2 \,\le\, \frac{1}{m} \sum_t y^t_i \,=\, \mu,
	\end{align*}
	where the inequality uses $y^t_i \in [0,1]$. Applying Lemma~\ref{lemma:maximalBernstein} with $$\alpha := \e m \mu + \frac{2 \cdot (24)^2}{\e} \, (\log \nicefrac{d}{\delta} + \log 30)$$ we get, since $\alpha^2 \ge 4 m \mu\, (24)^2\, (\log \nicefrac{d}{\delta} + \log 30)$,
 	\begin{align*}
 	 \Pr\left(\max_{j \le m/2} |{\textstyle \sum_{t\le j}} Y^j_i - j\mu| \ge \alpha \right) &\le 30 \exp\left(-\min\bigg\{\frac{4 m \mu (\log \nicefrac{d}{\delta} + \log 30)}{2 m \mu}~,~\frac{2 (24)^2 (\log \nicefrac{d}{\delta} + \log 30)/\e}{48}\bigg\}\right) \\
 	 &\le 30 e^{-(\log d/\delta + \log 30)} ~\le~ \frac{\delta}{d}.
 	\end{align*} 

	Finally, whenever this event holds, equation \eqref{eq:expOhat} gives that for all $j \le m/2$
	\begin{equation*}
	\E_{j-1} Y^j_i \le \frac{(m - (j-1)) \,\mu + \alpha}{m - (j-1)} \le \mu + \frac{\e m \mu + O(\frac{\log \nicefrac{d}{\delta}}{\e})}{m - (j-1)} \le (1+2\e) \mu + \frac{O(\log d/\delta)}{m \e},
	\end{equation*}	 
	the last inequality using $j \le \frac{m}{2}$. This concludes the proof. 
	\end{proof}

	\begin{proof}[Proof of Lemma \ref{lemma:breakCorr}]
		Since $Z^t$ and $Y^j$ are independent conditioned on $\cF_{j-1}$, we have
		\begin{align*}		
		\E_{j-1} \ip{Z^j}{Y^j} \,=\, \ip{\E_{j-1} Z^j}{\E_{j-1} Y^j} 
		\end{align*}
		Moreover, applying a union bound on Lemma \ref{lemma:balancedEveryTime} over all coordinates $i$, with probability at least $1-\frac{\delta}{2}$ for all $j \le \frac{m}{2}$ (in particular for all $j \le \tau$) we have $\ip{\E_{j-1} Z^j}{\E_{j-1} Y^j} \le (1+2\e) \ip{\E_{j-1} Z^j}{\E Y^j} + \frac{O(\log d/\delta)}{m \e}$. Adding over all $j \le \tau$ we get that
		\begin{flalign}
		&& \sum_{j \le \tau} \E_j \ip{Z^j}{Y^j} &\le (1+2\e) \sum_{j \le \tau} \ip{\E_{j-1} Z^j}{\E Y^j} + \frac{O(\log d/\delta)}{\e} &&\textrm{with probability $\ge 1 - \frac{\delta}{2}$}\,. \label{eq:condBound}
		\end{flalign}

		We now show using Lemma \ref{lemma:nikhil} that with good probability the desired quantity $\sum_{j \le \tau} \ip{Z^j}{Y^j}$ is close to $\sum_{j \le \tau} \E_j \ip{Z^j}{Y^j}$. Define the stopped random variable 
		\begin{align*}
		X_j := \one(\tau \ge j) \cdot \Big[(1-\e) \ip{Z^j}{Y^j} - \E_j \ip{Z^j}{Y^j}\Big].
		\end{align*}
		Recall that by definition of stopping time, the event $\one(\tau \ge j)$ is $\cF_{j-1}$-measurable, and hence $\E_j X_j = \one(\tau \ge j) \cdot (-\e\, \E_j \ip{Z^j}{Y^j})$. Moreover,
		\begin{align*}
			\E_j X^2_j ~&=~ \one(\tau \ge j) \cdot \Big[(1-\e)^2 \,\E_j \underbrace{\ip{Z^j}{Y^j}^2}_{\le \ip{Z^j}{Y^j}} \,-\, \underbrace{(2 (1-\e) - 1)}_{\ge 0} (\E_j \ip{Z^j}{Y^j})^2\Big] \\
			~&\le~ \one(\tau \ge j) \cdot \E_j \ip{Z^j}{Y^j},
		\end{align*}  
		where the first underbrace is because $\ip{Z^j}{Y^j} \le 1$ and the second because $\e \in (0, \frac{1}{2}]$. Together, these observations give $$\E_j X_j \le - \e\, \E_j X_j^2.$$ Then applying Lemma \ref{lemma:nikhil} to the sequence $(X_j)_j$ with $\lambda = \frac{\log \nicefrac{1}{2\delta}}{\e}$ we obtain
		\begin{align*}
		\Pr\bigg((1-\e) \sum_{j \le \tau} \ip{Z^j}{Y^j} ~>~ \sum_{j \le \tau} \E_j \ip{Z^j}{Y^j} + \frac{\log \nicefrac{1}{2\delta}}{\e}\bigg) &~\le~ \frac{\delta}{2}.
		\end{align*}
		
		Then by union bound with \eqref{eq:condBound}, with probability at least $1- \delta$ we have
		\begin{align*}
			\sum_{j \le \tau} \ip{Z^j}{Y^j} ~\le~ \frac{(1+2\e)}{(1-\e)}\, \sum_{j \le \tau} \ip{\E_{j-1} Z^j}{\E Y^j} + \frac{O(\log \nicefrac{d}{\delta})}{\e}.
		\end{align*}
		Verifying that $\frac{(1+2\e)}{(1-\e)} \le 1+4\e$ for all $\e \in (0,\frac{1}{10}]$ then proves Lemma \ref{lemma:breakCorr}.
		
	\end{proof}

\section{Missing Proofs from \Cref{sec:byzant-PIP}}
\label{sec:appendix-pip}
  
\Reduction*

\begin{proof}
	The desired $O(\rho)$-approximation for general instances is given by choosing uniformly at random and running one of the following 3 algorithms:

	\begin{enumerate}
		\item Pick one of the $n$ items uniformly at random
		
		\item Run the algorithm given by \Cref{lem:itcs-result} aiming 
		at picking the best $B$ items
		
		\item See the largest value $c^{1/2}_{\max}$ of an item with arrival time in $[0,\frac{1}{2})$ (do not pick any) and on the remaining times $[\frac{1}{2}, 1]$ run an algorithm that is $\rho$-competitive in a smooth instance with constant probability using $c^{1/2}$ as estimate for $\OPT(\Gitems)$ (still using budget $B$). 
	\end{enumerate}	

	By construction this procedure always produces a feasible solution, and we show it has good value in expectation.  
	
	Let $c_{\max}$ be the maximum value of over all (green and red) items. First, if $c_{\max} \ge \frac{n}{2} \OPT(\Gitems \setminus g_{\max})$ then the procedure has expected value at least $\frac{1}{3} \OPT(\Gitems)$ just from the first algorithm that it may run, and the result follows. Also, if the $B$ top valued green items (excluding $g_{\max}$) have combined value at least $\frac{\OPT(\Gitems \setminus g_{\max})}{4}$, then the procedure gets expected value at least $\Omega(\OPT(\Gitems \setminus g_{\max}))$ just from the second algorithm that it may run, and the result also follows.

	So consider the ``remaining situation'' where neither of these cases happen. Further, condition on the event where $g_{\max}$ shows up at a time $[0,\frac{1}{2})$ and $\OPT(\Gitems \cap [\frac{1}{2}, 1])$ (the optimal value considering only green items on times $[\frac{1}{2},1]$) it at least $\frac{1}{2} \OPT(\Gitems \setminus g_{\max})$, which happens with probability at least $\frac{1}{4}$. In this case the instance over times $[\frac{1}{2}, 1]$ satisfies both items of the smoothness Assumption \ref{assum:smooth} with $\widehat{O} := c^{1/2}_{\max}$ because:

	\begin{enumerate}
		\item Item 1: The $B$ top valued green items in $[\frac{1}{2},1]$ have combined value at most $\frac{\OPT(\Gitems \setminus g_{\max})}{4} \le \frac{\OPT(\Gitems \cap [\frac{1}{2}, 1])}{2}$. Since this set of items contains all items in $[\frac{1}{2},1]$ of value at least $\frac{\OPT(\Gitems \cap [\frac{1}{2}, 1])}{B}$, it satisfies Item 1 of the assumption. 
		
		\item Item 2: Using $c(g_{\max})$ to denote the value of $g_{\max}$,
		\begin{align*}
			c^{1/2}_{\max} &~\ge~ c(g_{\max}) ~\ge~ \frac{1}{n} \OPT(\Gitems) ~\ge~ \frac{1}{n} \OPT(\Gitems \cap [\tfrac{1}{2}, 1])\\
			c^{1/2}_{\max} &~\le~ \frac{n}{2} \OPT(\Gitems \setminus g_{\max}) ~\le~ n \OPT(\Gitems \cap [\tfrac{1}{2}, 1]).
		\end{align*} 		
	\end{enumerate}
	Thus, under this event, with probability $\frac{1}{3}$ the third algorithm within the procedure is run and with further constant probability it is guaranteed to obtain expected value at least $\Omega(\rho \OPT(\Gitems \cap [\frac{1}{2}, 1])) \ge \Omega(\rho \OPT(\Gitems \setminus g_{\max}))$. Overall, in this ``remaining situation'' the procedure obtains expected value at least $\Omega(\rho\OPT(\Gitems \setminus g_{\max}))$, thus concluding the proof.	\end{proof} 

      \subsection{Proof of \Cref{lemma:LOPT}} \label{sec:lemma:LOPT}

Recall that  $\stoch(I) \subseteq \N$ denotes the steps $t$ where the $t$-{th} item in the interval $I$ is
green.
We first argue that the value of the green items in any interval $I$ is large, with high probability. Recall that $x^*$ is the optimal solution consisting only of green items each of whose value is at most $\frac{OPT}{B}$, and has total value at least $\frac{OPT}{2}$.

\begin{claim} \label{claim:stochVal1}
If $|I|=1/\parts$ and $B \geq \parts\log 1/\delta$, then with probability at least $1-\frac{\delta}{2}$ we have
\[		
		\sum_{t \in \stoch(I)} C_t x_t^* ~\geq ~ \frac{1}{4} \frac{OPT}{\parts}.
\]
\end{claim}

\begin{proof}
  The LHS is $\sum_{i \in \Gitems} c_i x^*_i \cdot \one(i \in I)$ and
  has expectation at least $\frac{1}{2}\frac{\OPT}{\parts}$ and variance
  \begin{align*}
    \Var\left(\sum_{i \in \Gitems} c_i x^*_i \cdot \one(i \in
    I)\right) ~\le~ \frac1\parts \sum_{i \in \Gitems} \big(c_i x_i^*\big)^2 ~\le~ \frac{\OPT}{B\parts}\, \sum_{i \in \Gitems} c_i x_i^* ~\le~ \,\frac{\OPT^2}{B \parts} \enspace ,
  \end{align*}
  where the second inequality uses that $x_i^*\in[0,1]$ and that, by definition, $x_i^* > 0$ only when item $i$ has value $c_i \le \frac{OPT}{B}$.
  Then applying Bernstein's Inequality (\Cref{lem:Bern}) to $X_i := c_i x_i^* (\one(i\in I) - |I|)$,
  \begin{align*}
    \Pr\bigg(\sum_{i \in \Gitems} c_i x^*_i \cdot \one(i \in I) \le
    \frac{1}{2}\frac{\OPT}{\parts} - \frac{1}{4} \frac{\OPT}{\parts}\bigg) 
    ~\leq~ 2e^{-\frac{3B}{64\parts}} \enspace,
  \end{align*}	
  and the result follows from the assumption
  $B \ge \Omega(\parts \log 1/\delta)$.
\end{proof}

Next we argue that the total cost in the Lagrangified value is 
not  large. The proof of this claim goes by first conditioning on the stochastic times of
the green items inside interval $I$. 
This fixes the   order in which we see all the items, and also fixes the number of greens
  in $I$. Now at each stochastic time in $I$, we draw an
  item from the remaining greens (without replacement), and
   use \Cref{lemma:breakCorr} 
 to bound the effect of sampling without replacement.

\begin{claim} \label{claim:stochVal2} If $|I| = \frac{1}{\parts} \le \frac{1}{4}$, and $B \ge \Omega(\parts\log (4d/\delta))$, then with
  probability at least $1-\frac{\delta}{2}$,
  $$\sum_{t \in \stoch(I)} \ip{\dual_t}{ A_t x^*_t} ~\le~ \frac{4B}{\parts} \enspace .$$
\end{claim}

\begin{proof}
  Recall that $\Gitems(I)$ and $\Gsize(I)$ denote the set and the number 
  of green items, respectively, that fall in interval
  $I$. Let $t_j$ be the position of the $j^{th}$ green item to
  appear in  interval $I$. Then
  $\sum_{t \in \stoch(I)} \ip{\dual_t}{ A_t x^*_t} = \sum_{j
    \leq \Gsize(I)} \ip{\dual_{t_j}}{ A_{t_j}
    x^*_{t_j}}$. Notice that the $t_j$'s are random.
  
  We condition on $\Gsize(I) = k$, say, and on their positions
  $(t_1,\ldots,t_k) =: t_{\le k}$. Notice that under this conditioning
  the stochastic items $(C_{t_j}, A_{t_j})$ are still sampled without
  replacement from the green items. Moreover, notice that
  $\dual_{t_j}$ is a function of the green items before the
  $j^{th}$ green item in the interval, and therefore on
  $(C_{t_1}, A_{t_1}), \ldots, (C_{t_{j-1}}, A_{t_{j-1}})$.
  (It also depends on the
  red items in $I$, but these are deterministic.) 
  
 	In order to upper bound $\sum_{j \leq \Gsize(I)} \ip{\dual_{t_j}}{ A_{t_j} x^*_{t_j}}$ with high-probability (conditioned on $\Gsize(I) = k$ and $t_{\le k}$) we use \Cref{lemma:breakCorr}. For that, we first notice that due to the feasibility of $x^*$ we have $\E [A_{t_j} x^*_{t_j} \mid \Gsize(I) = k,\, t_{\le k} ] \le \frac{B \cdot \ones}{\Gsize} $ for all $j$, and hence  	
  \begin{align*}
  \sum_{j \le k} \Bigip{\lambda_{t_j}}{\E\big[A_{t_j} x^*_{t_j} ~\big|~ G(I) = k, t_{\le k}\big]} \le \sum_{j \le k} \ip{\lambda_{t_j}}{\tfrac{B \cdot \ones }{\Gsize}} = \frac{B}{G} k.
  \end{align*}
  Then applying \Cref{lemma:breakCorr} conditionally (setting $Z^j := \lambda_{t_j}|_{\Gsize(I)=k, t_{\le k}}$, $Y^j := A_{t_j} x^*_{t_j}|_{\Gsize(I)=k, t_{\le k}}$, and $\e = \frac{1}{10}$) gives that if $k \le \frac{\Gsize}{2}$ then
  \begin{align*}
    \Pr\bigg( \sum_j \ip{\dual_{t_j}}{ A_{t_j} x^*_{t_j}} \ge  \frac{3B}{2\Gsize}\,k\,+\,\Omega(\log \nicefrac{4d}\delta) ~\bigg|~ \Gsize(I) = k,\, t_{\le k} \bigg)~\le~\nicefrac\delta4 \enspace.
  \end{align*}
  Taking expectation over the $t_{\le k}$'s and then over the
  values of $k$ that are at most $\frac{2\Gsize}{\parts}$ (which is at most
  $\nicefrac{\Gsize}{2}$ by our assumption on $\parts \geq 4$), we get
  \begin{align}
    \Pr\bigg( \sum_j \ip{\dual_{t_j}}{ A_{t_j} x^*_{t_j}} \ge
    \frac{3 B}{K} \,+\,\Omega(\log \nicefrac{4d}\delta)
    ~\bigg|~ \Gsize(I) \le \nicefrac{2\Gsize}{\parts}
    \bigg)~\le~\nicefrac\delta4 \enspace. \label{eq:GM1}
  \end{align}	
  Observe that $\nicefrac{3B}{\parts} +  \Omega(\log \nicefrac{4d}\delta\}) \le \nicefrac{4B}{\parts}$
  by our assumption $B \ge \Omega(\parts\log (4d/\delta))$. 
  Next we show that the conditioning holds with high
  probability. Indeed, $\E [\Gsize(I)] = \nicefrac{\Gsize}{\parts}$, so by Bernstein's
  inequality
  \begin{align}
    \Pr(\Gsize(I) > \nicefrac{2\Gsize}{\parts})~\le~
    e^{-\frac{\Gsize}{2\parts}} ~\le~ e^{-\frac{B}{4\parts}} ~\le~
    \nicefrac{\delta}{4d} ~\le~ \nicefrac\delta4 \enspace, \label{eq:GM2}
  \end{align}
  where the second inequality uses Assumption \ref{assum:smooth} that the green items with
  value at most $\frac{\OPT}{B}$ contain a solution of value at least $\frac{\OPT}{2}$ (hence there are at least $\frac{B}{2}$ green items to obtain the remaining value $\frac{\OPT}{2}$) 
     and the third inequality uses the assumption that $B \ge \Omega(\parts \log
  (4d/\delta))$. Combining~\eqref{eq:GM1} and \eqref{eq:GM2},
  %
  \begin{align*}
    \Pr\bigg( \sum_j \ip{\dual_{t_j}}{ A_{t_j} x^*_{t_j}} \ge
    \frac{4 B}{\parts} \bigg)~\le~\delta/2 \enspace,
  \end{align*}	
  which completes the proof of Claim \ref{claim:stochVal2}.
\end{proof}
Finally, using Claims~\ref{claim:stochVal1} and~\ref{claim:stochVal2} and taking a
union bound concludes the proof of
\Cref{lemma:LOPT}.

\section{Missing Proofs from \Cref{sec:prophets}}

\subsection{Proof that Assumption \ref{assum:smoothProphet} is WLOG}\label{sec:redSmoothProphet}

The following lemma formalizes the idea that
Assumption \ref{assum:smoothProphet} can be made without loss of generality.

\begin{lemma}
Let $\widetilde{\alg}$ be an algorithm for packing linear programs in the Prophets with Augmentations model. Suppose that $\widetilde{\alg}$ achieves expected value at least $\Omega(\OPT_{base})$ on instances where each distribution is supported on values that are at most $\frac{\OPT_{base}}{20}$. Then there is an algorithm that achieves expected value $\Omega(\OPT_{base})$ on arbitrary instances.
\end{lemma}

\begin{proof}
\newcommand{\trunc}{\textsf{trunc}}
Throughout this section, we use $V_1,\ldots,V_n$ to denote the outcome of the value of the items of the original instance $((a_t, \cD_t)_t, B)$, $R_1,\ldots,R_n$ as the augmentations performed by the adversary, and $C_t = V_t+R_t$ the final value revealed to the algorithm. 

For notational convenience let $M := \frac{\OPT_{base}}{40}$.
The idea is to run two algorithms, one over the items that have value at most $M$ and one over items of value above $M$. To make this precise, define the operation $\trunc(v) := \min\{v, M\}$ that truncates a value at $M$. Let $\widetilde{\cD}_t$ be the distribution of the truncated random variable $\trunc(V_t)$. Then we consider the algorithm that flips an unbiased coin runs one of the following procedures on the (augmented version of) the original instance $((a_t, \cD_t)_t, B)$:

	\begin{enumerate}
	\item $\alg_{low}$: It sends the chopped instance $((a_t, \widetilde{\cD})_t, B)$ to the algorithm $\widetilde{\alg}$ to obtain a selection policy, and apply this policy to the sequence of truncated values $\trunc(C_1),\ldots,\trunc(C_n)$ to decide which items to take. Let $X_t \in \{0,1\}$ denote the indicator whether this policy picked the $t^{th}$ item.
	
	\item $\alg_{high}$: It picks the first items with value $C_t$ above $M$, if any. Let $\tau \in [n]$ be the index of the item picked ($\tau = \infty$ if did not pick any).
	\end{enumerate}

	We claim that either $\alg_{low}$ or $\alg_{high}$ has value at least $\Omega(\OPT_{base})$, which then proves the lemma. For that, let $p := \Pr(\tau < \infty)$ be the probability that some item has value above $M$. If $p \ge \frac{1}{2}$, then $\alg_{high}$ already has expected value least $p M \ge \Omega(\OPT_{base})$. We henceforth assume that $p < \frac{1}{2}$.

	Let $\OPT_{trunc}$ be the optimal value of the truncated instance $((a_t, \widetilde{\cD})_t, B)$, namely $$\OPT_{trunc} = \E \max_x\bigg\{ \sum_t \trunc(V_t) \cdot x_t ~:~ \sum_t a_t x_t \le B \cdot \ones\,,\,x \in \{0,1\}^n\bigg\}.$$ We consider two cases:
	
	\paragraph{Case 1: $\OPT_{trunc} \ge \frac{\OPT_{base}}{2}$.} In every scenario the value of $\alg_{low}$ is
	\begin{align}
		\alg_{low} = \sum_t C_t X_t \ge \sum_t \trunc(C_t) \cdot X_t. \label{eq:redSmoothPro1}
	\end{align}
	Moreover, notice that the sequence $\trunc(C_1),\ldots,\trunc(C_n)$ can be seen as an augmented version of the truncated instance $((a_t, \widetilde{\cD})_t, B)$, where the adversary performed the augmentation $\widetilde{R}_t := \trunc(C_t) - \trunc(V_t)$ that is non-negative and only depends on $V_t$. Moreover, in this Case 1 we have that all values in the truncated instance are at most $M = \frac{\OPT_{base}}{40} \le \frac{\OPT_{trunc}}{20}$. Therefore, the approximation guarantee of $\widetilde{\alg}$ holds in this case hence 
	\begin{align*}
		\E \sum_t \trunc(C_t)\cdot X_t \ge \Omega(\OPT_{trunc}) \ge \Omega(\OPT_{base}),
	\end{align*}
	where the last inequality again uses the assumption of Case 1. Combined with \Cref{eq:redSmoothPro1} this gives that $\alg_{low}$ has expected value at least $\Omega(\OPT_{base})$ as desired.

	\paragraph{Case 2: $\OPT_{trunc} < \frac{\OPT_{base}}{2}$.} Let $X^* \in \{0,1\}^n$ be the optimal solution for the original base instance $((a_t, \cD_t)_t, B)$, namely $\OPT_{base} = \E \sum_t V_t X^*_t$. Since $$V_t \,=\, \trunc(V_t) + (V_t - \trunc(V_t))\cdot \one(V_t > M),$$ we get 
	\begin{align*}
		\OPT_{base} &= \E \sum_t \trunc(V_t)\cdot X^*_t + \E \sum_t (V_t - \trunc(V_t))\cdot \one(V_t > M) \cdot X^*_t\\
		&\le \OPT_{trunc} + \E \sum_t V_t \cdot \one(V_t > M)\\
		& = \OPT_{trunc} + \sum_t \E[V_t \mid V_t > M]\,\Pr(V_t > M),
	\end{align*}
	where the inequality follows from the fact that $X^*$ is a feasible solution for $\OPT_{trunc}$. Moreover, since we are in Case 2, the second term in the RHS must contribute to at least half of $\OPT_{base}$, namely
	\begin{align}
		\OPT_{base} \,&\le~ 2 \sum_t \E[V_t \mid V_t > M]\,\Pr(V_t > M). \label{eq:redSmoothPro}
	\end{align}
		
	Now notice that we can express the value of $\alg_{high}$ in every scenario as 
	\begin{align*}
		\alg_{high} = C_\tau \ge V_\tau = \sum_t \one(\tau \ge t) \cdot \one(V_t > M) \cdot V_t,
	\end{align*}
	so using the fact that $\tau$ is a stopping time (so $\one(\tau \ge t)$ is defined by the history up to time $t-1$) and that the $V_t$'s are independent, the expected value becomes
	\begin{align*}
		\E \alg_{high} = \sum_t \Pr(\tau \ge t) \cdot \E[V_t \mid V_t > M]\,\Pr(V_t > M) \ge \frac{1}{2} \sum_t \E[V_t \mid V_t > M]\,\Pr(V_t > M),
	\end{align*}
	the inequality following because the probability of $\tau < t$ is at most the probability $p$ that any item takes value above $M$, and because we have assumed $p \le \frac{1}{2}$. Comparing with \eqref{eq:redSmoothPro} we see that $\E \alg_{high} \ge \Omega(\OPT_{base})$ as desired. This concludes the proof. 
\end{proof}

\subsection{Proof of \Cref{lem:prophBoundOPT}}
\label{sec:non-robust-prophet}

  We consider the following the concave relaxation to the base prophet
  instance, where intuitively $x_t$ denotes the fraction of times item
  $t$ is picked by the optimal offline algorithm:
  \begin{align*}
    \max_{x_1,\ldots, x_n} ~&\sum_t x_t \cdot \E[V_t \mid V_t \textrm{ is in its top $x_t$-quantile of $\cD_t$}]\\
    \text{s.t.}~~& \sum_t a_t x_t \le  \nicefrac{B}{4} \cdot \ones\\
                            & 0\leq x_t  \leq 1 \enspace .
  \end{align*}  
  We assume that the distributions are continuous, so
  that the quantiles are well-defined; this is without loss of
  generality (see, e.g., \cite[\S2]{RubinsteinWW20}).
  
  Let $x_t^*$ denote the optimal solution of the relaxation, and let
  $\val_t := \E[V_t \mid V_t \textrm{ is in its top
    $x_t^*$-quantile}]$.  To prove that this relaxation's objective
  value $\sum_t x_t^* \, \val_t$ is at least $\frac{\OPT_{base}}{4}$, observe that
  if an item $t$ is picked for $x_t$ fraction of times by the offline optimal
  solution, its contribution to the offline objective is at most
  $\E[V_t \mid V_t \textrm{ is in its top $x_t$-quantile}]$.  Note
  that although the relaxation only allows budget $\nicefrac{B}{4}\cdot \ones$
  (instead of $B\cdot \ones$), this only hurts the relaxation's
  objective by at most a factor of $4$.

  Next we design the desired solution $\psi_1(V_1),\ldots,\psi_n(V_n)$ for the base prophet instance
  with expected: We pick item $t$ whenever its value $V_t$ is in the top $x_t^*$-quantile of $\cD_t$, that is, $\psi_t(V_t) = \one(V_t$ is in the top $x_t^*$-quantile of $\cD_t)$.
  The expected budget consumed by such an algorithm is at most
  $\sum_t a_t x_t^* \leq \nicefrac{B}{4} \cdot \ones$, which proves
  \Cref{eq:budgetOPTbase}. To prove \Cref{eq:relaxBndOPT}, note that 
  the algorithm's expected value is precisely the objective value $\sum_t x_t^* \, \val_t$ of the 
  relaxation above, which is at least $\frac{\OPT_{base}}{4}$. 


{\small
\bibliographystyle{alpha}
\bibliography{bib,online-lp-short}
}

\end{document}